\documentclass[a4paper]{article}

\title{Radon-based Image Reconstruction for MPI using a continuously rotating FFL}
\author{Stephanie Blanke\footnote{stephanie.blanke@uni-hamburg.de}~ and Christina Brandt}
\date{}

\usepackage{graphicx}        

\usepackage{subcaption}

\usepackage{amsfonts,amsmath}
\usepackage{amssymb}

\usepackage{enumitem}

\usepackage[pdftex,colorlinks, citecolor=black,linkcolor=black,urlcolor=black]{hyperref}

\usepackage{color}

\usepackage[a4paper, textwidth=14cm]{geometry}

\DeclareMathOperator{\R}{\mathbb R}
\DeclareMathOperator{\dmathInt}{\, d}

\DeclareMathOperator{\BV}{BV}
\DeclareMathOperator{\TV}{TV}
\DeclareMathOperator{\supp}{supp}

\newcommand{\brackets}[1]{\left( #1 \right)}
\newcommand{\abs}[1]{\left| #1 \right|}
\newcommand{\norm}[1]{\left\| #1 \right\|}

\newtheorem{definition}{Definition}[section]
\newtheorem{lemma}[definition]{Lemma}
\newtheorem{remark}[definition]{Remark}
\newtheorem{example}[definition]{Example}
\newtheorem{theorem}[definition]{Theorem}

\newenvironment{proof}{\textbf{Proof:}}{\hfill$\square$ \\}

\begin{document}

\maketitle

\vspace{-0.7cm}

\begin{center}
	\small
	Department of Mathematics, Universit\"at Hamburg, Germany
\end{center}

\vspace{0.05cm}

\begin{abstract}
Magnetic particle imaging is a relatively new tracer-based medical imaging technique exploiting the non-linear magnetization response of magnetic nanoparticles to changing magnetic fields. If the data are generated by using a field-free line, the sampling geometry resembles the one in computerized tomography. Indeed, for an ideal field-free line rotating only in between measurements it was shown that the signal equation can be written as a convolution with the Radon transform of the particle concentration. In this work, we regard a continuously rotating field-free line and extend the forward operator accordingly. We obtain a similar result for the relation to the Radon data but with two additive terms resulting from the additional time-dependencies in the forward model. We jointly reconstruct particle concentration and corresponding Radon data by means of total variation regularization yielding promising results for synthetic data.
\end{abstract}

\section{Introduction}
Reliable and fast medical imaging techniques are indispensable for diagnostics in clinical everyday life. One promising example is magnetic particle imaging (MPI) invented by Gleich and Weizenecker \cite{gleich2005tomographic}. MPI is a tracer-based imaging modality aiming for the reconstruction of the spatial distribution of magnetic nanoparticles injected into the patient's body. For data generation changing magnetic fields are applied to the field of view (FOV). These fields are composed of a selection field $\mathbf{H}_S$ featuring a low-field volume (LFV), and a drive field $\mathbf{H}_D$ moving the LFV through the FOV. The particles' corresponding non-linear magnetization response is measured in terms of induced voltages via receive coils. Note that only particles within the LFV contribute to the signal as all other particles stay in magnetic saturation, guaranteeing spatial information to be encoded in the data. The LFV takes different shapes depending on the specific scanner implementation. Most commonly, a field-free point (FFP) is used for spatial encoding. However, alternatively a field-free line (FFL) can be used as it was first suggested by~\cite{weizenecker2008magnetic}. For a detailed introduction to MPI, we recommend to consult~\cite{knopp2012magnetic}. For the specific case of an FFL scanner we further refer to~\cite{bringout2016field} and~\cite{erbe2014field}. \\
For the remainder, we restrict ourselves towards the FFL encoding scheme. The first FFL imaging systems were developed by~\cite{goodwill2012projection} and~\cite{bente2014electronic}. For further steps and breakthroughs in MPI history we refer to~\cite{knopp2017magnetic}. In~\cite{mattingly2020mpi}	the open-source project OS-MPI is presented. Thereby, information about a small-bore FFL imager is shared. Advantages of using an FFL scanner lie in a possible increase in sensitivity~\cite{weizenecker2008magnetic} as more particles are contributing to the signal. Further, in~\cite{kluth2018degree} it is stated that using an FFL may lead to a less ill-posed problem compared to the FFP scanner.\\
During data acquisition the FFL is rotated and translated through the FOV resulting in scanning geometries resembling those in computerized tomography (CT). In CT the intensity loss of X-rays traversing the object under investigation is measured. Having access to data for a suitable choice of different positions for radiation source and detector panel, the attenuation coefficient of the specimen can be reconstructed. The forward operator is given by the Radon transform mapping a function to the set of its line integrals. See e.g.~\cite{natterermathematics} for more information concerning CT. Since only particles in close vicinity to the FFL contribute to the induced voltage, it stands to reason that the MPI signal equation is linked to the Radon transform of the particle concentration. Indeed, it was shown in~\cite{knopp2011fourier} that MPI data can be traced back to the Radon transform of the particle concentration. In~\cite{bringout2020new} a similar result was derived using their newly developed 3D model applicable to magnetic fields approximately parallel to their velocity field. Hence, for concentration reconstruction, results and methods from the well-known imaging technique CT are accessible. Because of the aforementioned analogies, the idea aroused to combine the MPI with a CT scanner, thus combining quantitative information about the tracer distribution with structural information about the tissue itself. The first hybrid MPI-CT scanner was proposed in~\cite{vogel2019magnetic}. \\
For the derivation of the relation between Radon and MPI data in~\cite{knopp2011fourier}, the FFL is assumed to be rotated in between measurements. In this work, we will use a slightly different scanning geometry, allowing the FFL to be rotated simultaneously with its translation as presented in the initial publication with respect to MPI using a field-free line~\cite{weizenecker2008magnetic}. In~\cite{bringout2020new} and~\cite{knopp2011fourier} it was proposed that the derived relation is still applicable in case that the rotation is sufficiently slow compared to the translation speed. In the remainder sections, we will investigate the supposition in some more details. \\
Throughout the article we will apply the Langevin theory to model the forward operator. This corresponds to assuming the particles to be in thermal equilibrium~\cite{knopp2012magnetic}. While being not a proper description of the particle dynamics, it is still the state of the art approach when using a model-based formulation of the signal equation. Alternatively, commonly a measurement-based approach is applied, where, before the actual data generation takes place, measurements are taken for a small probe being moved trough the FOV called calibration process~\cite{knopp2012magnetic}. However, in recent years the problem of finding more accurate modeling methods aroused further interest, cf. (\cite{kluth2019towards},~\cite{kaltenbacher2021parameter},~\cite{weizenecker2018fokker}) to name a few examples. In \cite{kaltenbacher2021parameter} e.g. the modeling task was traced back to an inverse parameter identification problem. In~\cite{kluth2018mathematical} a survey regarding mathematical modeling of the signal chain is given.\\   
Finally, inspired by~\cite{Tovey_2019} we propose an image reconstruction approach, which jointly solves for the particle concentration as well as its corresponding Radon data by means of total variation (TV) regularization. Utilizing TV-based methods has a broad range of applications and is especially suitable for piecewise constant functions. See e.g.~\cite{burger2013guide} for an introduction to TV regularization. Also in the setting of MPI TV has already been applied  (\cite{storath2016edge},~\cite{zdun2021fast},~\cite{bathke2017improved}). In~\cite{ilbey2017comparison} the system-based approach combined with TV regularization has been compared with the projection-based approach using the Radon transform  for an ideal FFL setup. \\\\
The article is organized as follows: In Section~\ref{Sec:FFL} we fix notation, introduce the MPI forward model for an FFL scanner, and review the link between MPI data and Radon transform of the particle concentration developed in~\cite{knopp2011fourier}. Section~\ref{Sec:RelMPIRadon} extends this relation to the setting of simultaneous line rotation. In Section~\ref{ImageReco} we propose a joint reconstruction of particle concentration and corresponding Radon data applying TV regularization. We close with numerical examples for synthetic data in Section~\ref{Results}. \\\\

\section{Field-free Line Magnetic Particle Imaging} \label{Sec:FFL} 
According to~\cite{knopp2012magnetic} the particle concentration $c: \, \R^3 \to \R^+_{0} := \R^+ \cup \left\lbrace 0\right\rbrace $ can be linked to the voltage signal $u_l: \, \R^+ \to \R,\; l\in\left\lbrace 1,\dots,L\right\rbrace $ induced in the $l$-th receive coil via
\begin{equation}
u_l\brackets{t} = -\mu_0 \int_{\R^3} c\brackets{\mathbf{r}} \frac{\partial}{\partial t}\overline{\mathbf{m}}\brackets{\mathbf{r}, t} \cdot \mathbf{p}_l\brackets{\mathbf{r}} \dmathInt\mathbf{r}.
\label{SignalEquation}
\end{equation}
Thereby, $\mu_0$ denotes the magnetic permeability, $\overline{\mathbf{m}}: \, \R^3 \times\R^+ \to \R^3$ the mean magnetic moment, and $\mathbf{p}_l: \, \R^3 \to \R^3$ the receive coil sensitivities. Note that we omit signal filtering. In practice, the direct feed-through of the excitation signal needs to be removed. In the following, we will assume that $c$ and $\frac{\partial}{\partial t}\overline{\mathbf{m}}\brackets{\cdot, t} \cdot \mathbf{p}_l\brackets{\cdot}$ are $L_2$ functions, i.e. square-integrable. Thus,~\eqref{SignalEquation} is well-defined. When the field-free region of the MPI implementation is a straight line and the magnetic fields are constant along lines parallel to this FFL, it is called ideal FFL scanner. The attribute \textit{ideal} is added as in practice we are confronted with field imperfections leading to deformed LFVs~(\cite{bringout2016field},~\cite{bringout2020new}). In the following, we consider an ideal FFL scanner generating data for different positions and directions by moving the FFL through the xy-plane. For convenience, we will assume the particle concentration to be contained within this plane and regard the problem as two dimensional $c:\; \R^2 \to \R^+_{0} $ (cf.~\cite{kluth2018degree} for more details). Further, let $\supp \brackets{c} \subset B_R \subset \R^2$ with $B_R$ denoting the circle of radius $R>0$ around the origin. Corresponding to~\cite{knopp2011fourier} we model the magnetic fields as
\begin{equation}
\mathbf{H}(\mathbf{r},\varphi,t) = \mathbf{H}_S(\mathbf{r},\varphi) + \mathbf{H}_D(\varphi,t) = \left( -G \; \mathbf{r}
\cdot \mathbf{e}_{\varphi}
+ A\Lambda(t)\right) \mathbf{e}_{\varphi}
\label{Eqn:MagneticField}
\end{equation}
with $\mathbf{e}_\varphi:=\brackets{-\sin\varphi,\cos\varphi}^T$. Here, $G$ is the gradient strength determining the width of the LFV, $A$ denotes the drive peak amplitude, and $\Lambda$ is a periodic excitation function usually chosen to be sinusoidal. A simple computation shows that the FFL which is orthogonal to $\mathbf{e}_\varphi$ and with displacement $ s_{t}:= \frac{A}{G} \Lambda\brackets{t}$ to the origin builds the center of the LFV
\begin{equation}
\text{FFL}\brackets{\mathbf{e_{\varphi}}, s_t}:=\left\lbrace \mathbf{r} \in \R^2 : \; \mathbf{r} \cdot \mathbf{e_{\varphi}} = s_t \right\rbrace.
\end{equation}
The corresponding geometry is depicted in Figure~\ref{Fig:Geom}. Looking at this visualization, by replacing the FFL with X-rays, similarities between the scanning process for MPI using an FFL scanner and the well-known medical imaging technique CT become obvious. Note that, as it was already mentioned in~\cite{knopp2011fourier}, when comparing with the notation in common CT literature $\varphi$ needs to be shifted by $\frac{\pi}{2}$. This is because in CT the angle is usually measured from the axis to the orthogonal vector $\mathbf{e_{\varphi}}$ instead to the X-ray itself.
\begin{figure}[htbp]
	\centering
	\includegraphics[width=0.4\linewidth]{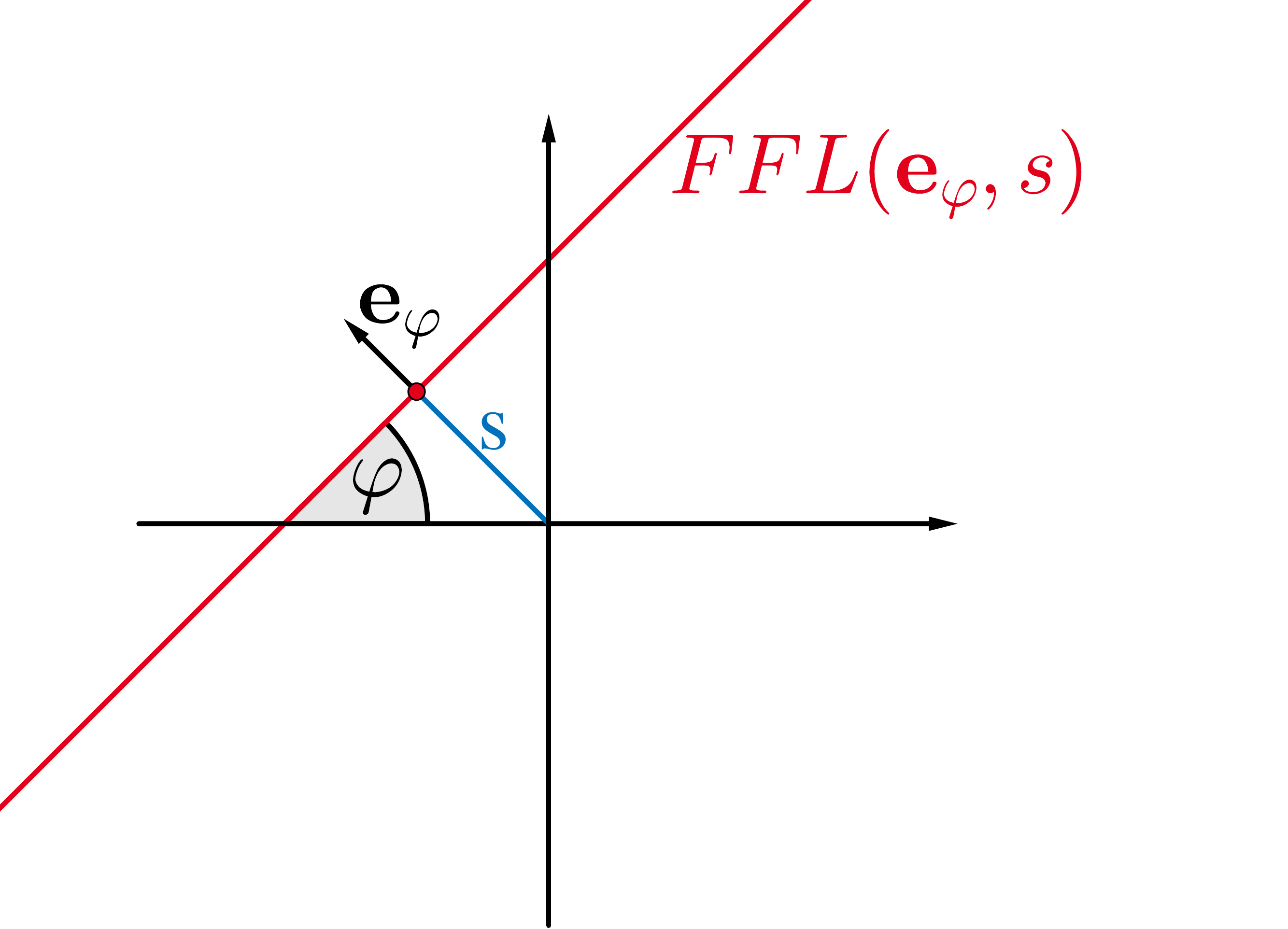}
	\caption{Visualization of the FFL orthogonal to $\mathbf{e}_\varphi$ and with displacement $s$ to the origin.}
	\label{Fig:Geom}
\end{figure}\\
Following~\cite{knopp2011fourier} the mean magnetic moment can be rewritten using the Langevin model of paramagnetism 
\begin{equation*}
\overline{\mathbf{m}}\brackets{\mathbf{r}, \varphi, t} = \overline{m}\brackets{\left\|  \mathbf{H}\brackets{\mathbf{r}, \varphi, t} \right\|} \frac{\mathbf{H}\brackets{\mathbf{r}, \varphi, t}}{\left\|  \mathbf{H}\brackets{\mathbf{r}, \varphi, t} \right\|}
\end{equation*}
with $\overline{m}\brackets{H}= m\mathcal{L}\brackets{\frac{\mu_0 m}{k_{\text{B}} T}H}$ denoting the modulus of the mean magnetic moment. Further, $m$ is the magnetic moment of a single particle, $k_{\text{B}}$ is the Boltzmann constant, $T$ the particle temperature, and $\mathcal{L}:\; \R \to \left[ -1,1\right] $ the Langevin function defined as 
\begin{eqnarray}
\mathcal{L}\brackets{\lambda} := \begin{cases}
\coth \brackets{\lambda} - \frac{1}{\lambda}&, \; \lambda \neq 0, \\
0&, \; \lambda=0.
\end{cases}\label{eq:Langevin}
\end{eqnarray}
Thus, $\overline{m}\brackets{-H}=-\overline{m}\brackets{H}$ and the signal equation~\eqref{SignalEquation} can be written as
\begin{equation}
u_l(\varphi, t) = -\mu_0  \int_{\R^2} c(\mathbf{r}) \frac{\partial}{\partial t} \overline{m} \left( 
-G \; \mathbf{r}
\cdot \mathbf{e}_\varphi
+ A\Lambda(t)\right) \mathbf{e}_\varphi \cdot \mathbf{p}_l\brackets{\mathbf{r}} \dmathInt\mathbf{r}.
\label{Eqn:ModelFFL}
\end{equation}
Different scanning geometries exist. We regard \textit{sequential line rotation}~(Figure~\ref{Fig:SeqRot}) i.e. the FFL is sequentially translated through the FOV and rotated in between measurements, and \textit{simultaneous line rotation}~(Figure~\ref{Fig:SimRot}) i.e. the FFL is rotated simultaneously to the translation (cf.~\cite{weizenecker2008magnetic},~\cite{knopp2011fourier}). In the case of simultaneous line rotation we need to replace $\varphi$ via $\varphi_t:=2\pi f_{\text{rot}}t$ with line rotation frequency $f_{\text{rot}}>0$ everywhere $\varphi$ occurs~\cite{bringout2020new}. 
\begin{figure}[htbp]
	\begin{subfigure}[b]{0.45\textwidth}
		\centering
		\includegraphics[width=1\linewidth]{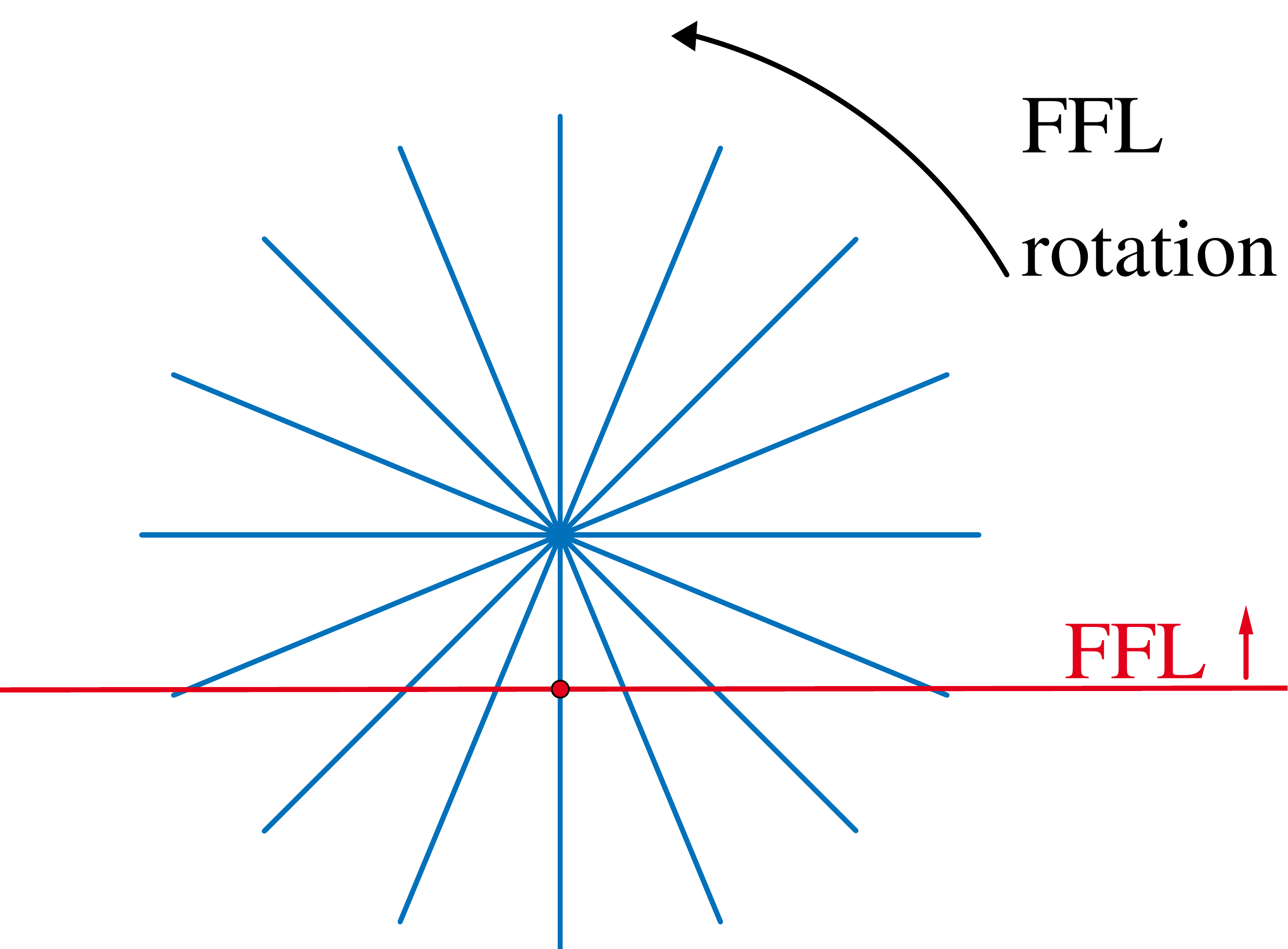}
		\subcaption{Sequential line rotation}
		\label{Fig:SeqRot}
	\end{subfigure}
	\hfill
	\begin{subfigure}[b]{0.45\textwidth}
		\centering
		\includegraphics[width=1\linewidth]{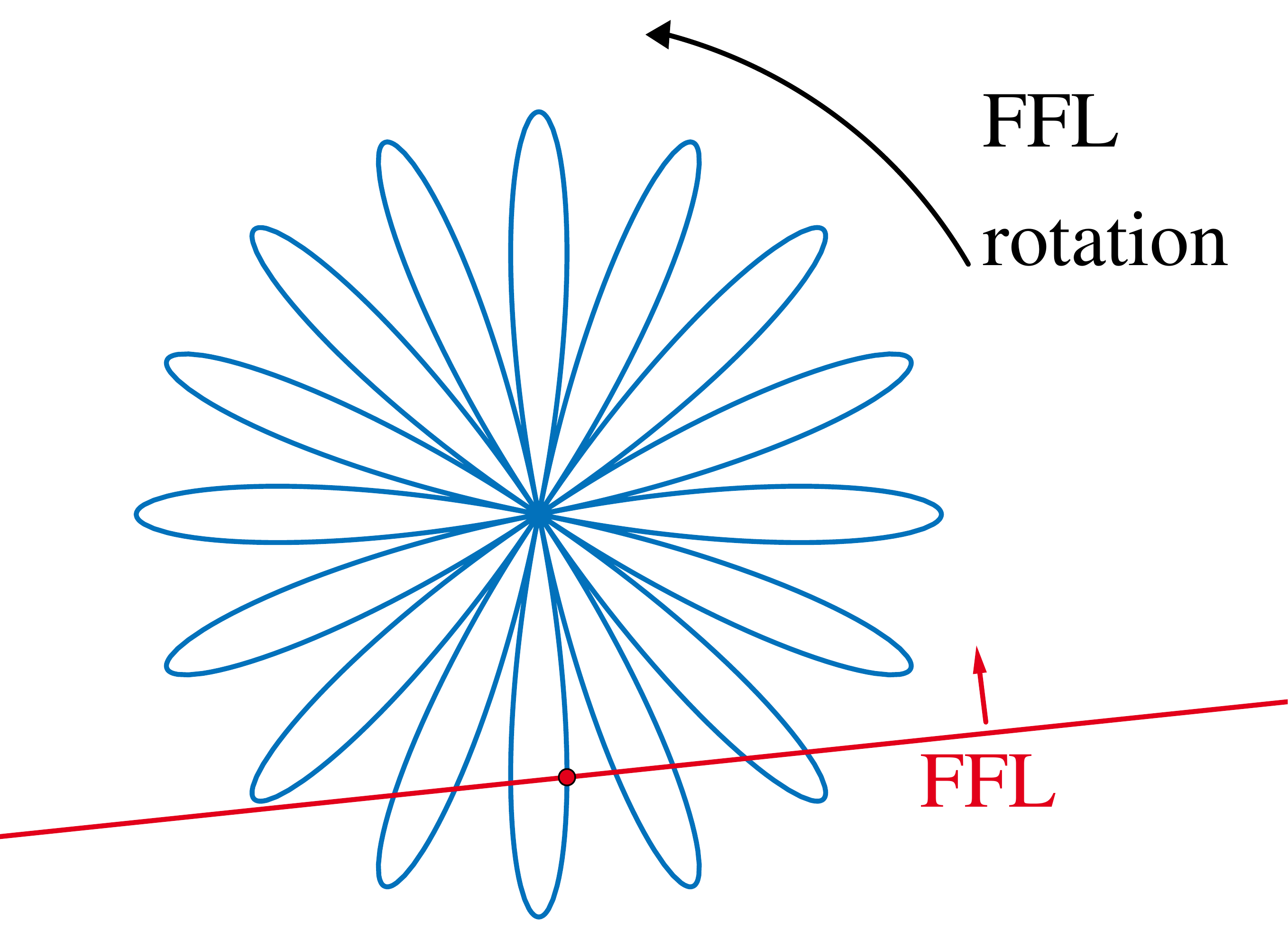}
		\subcaption{Simultaneous line rotation}
		\label{Fig:SimRot}
	\end{subfigure}
	\caption{Different scanning geometries: (a) Sequential translation and rotation of the FFL, (b) Simultaneous translation and rotation of the FFL.}
	\label{Fig:ScanningGeometry}
\end{figure} \\	
We now review the definition of the Radon transform. Let $S^1$ be the unit sphere in $\R^2$ and $Z := S^1  \times \left[ -R, R\right] $. The Radon transform $\mathcal{R}:\, L_2\brackets{B_R, \R^+_{0} } \to L_2\brackets{Z,\R^+_{0}  }$ of the particle concentration $c$ is given as
\begin{equation}
\mathcal{R}c\brackets{\mathbf{e}_\varphi, s} = \int_{B_R} c\brackets{\mathbf{r}} \delta\brackets{\mathbf{r}\cdot\mathbf{e}_\varphi-s} \dmathInt \mathbf{r}.
\label{Def:RadonTransform}
\end{equation} 
For sequential line rotation a link between MPI data an the Radon transform of the particle concentration was shown (see Theorem~\ref{Thm:FourierSlice} below).
\begin{theorem}{\cite{knopp2011fourier}} \label{Thm:FourierSlice}
	Given spatially homogeneous receive coil sensitivities, sequential line rotation, and ideal magnetic fields, the signal equation~\eqref{Eqn:ModelFFL} can be written as
	\begin{equation}
	u_l\brackets{\varphi, t} 
	= -\mu_0 A \Lambda'(t) \; \mathbf{e}_\varphi \cdot \mathbf{p}_l
	\left[  \overline{m}' \left(G \; \cdot\right) * \mathcal{R}c\brackets{\mathbf{e}_\varphi,\cdot} \right]  \left(\frac{A}{G}\Lambda\brackets{t}\right).
	\label{Eqn:FourierSlice}
	\end{equation}
\end{theorem}
For simultaneous line rotation it is reasonably proposed that~\eqref{Eqn:FourierSlice} holds approximately true as long as the FFL rotation is sufficiently slow compared to the translation~(\cite{knopp2011fourier},~\cite{bringout2020new}). The aim of this article is to investigate the setting of simultaneous line rotation in more detail. To this end, we define the MPI forward operator for the special case of ideal FFL scanner.
\begin{definition}\label{Def:MPI_FFL}
	Let $\mathcal{A}^{\text{FFL}}_l: \; L_2\brackets{B_R,\R^+_{0} } \to L_2\brackets{\R^+, \R}$ be defined as 
	\begin{equation}
	\mathcal{A}^{\text{FFL}}_l c\brackets{t} := -\mu_0  \int_{B_R} c(\mathbf{r}) \frac{\partial}{\partial t} \overline{m} \left( 
	-G \; \mathbf{r}
	\cdot \mathbf{e}_{\varphi_t}
	+ A\Lambda(t)\right) \mathbf{e}_{\varphi_t} \cdot \mathbf{p}_l\brackets{\mathbf{r}} \dmathInt\mathbf{r}.
	\label{Def:MPIForwardOp}
	\end{equation}
	Then, for ideal magnetic fields, the mapping $\mathcal{A}^{\text{FFL}}: \; L_2\brackets{B_R,\R^+_{0} } \to L_2\brackets{\R^+, \R^L}$ with $\mathcal{A}^{\text{FFL}} c\brackets{t} = \brackets{\mathcal{A}^{\text{FFL}}_lc\brackets{t}}_{l=1,\dots,L}$ is the forward operator for field-free line magnetic particle imaging.
\end{definition}
\begin{remark} \label{Rem:SimToSeq}
	Note that, by choosing $\varphi_t$ to be piecewise constant, sequential line rotation is also contained in the previous definition.
\end{remark}
For concentration reconstruction we need to solve the linear ill-posed inverse problem $$\mathcal{A}^{\text{FFL}} c = \mathbf{u} $$ with measured data $\mathbf{u} = \brackets{u_l}_{l=1,\dots,L}$ and forward operator $\mathcal{A}^{\text{FFL}}$.

\section{Relation of the MPI Forward Operator to the Radon Transform} \label{Sec:RelMPIRadon}
In the following, we examine the MPI forward operator in the ideal FFL setting for simultaneous line rotation. Therefore, let
\begin{eqnarray*}
	\mathbf{R}^{{\varphi}}=\begin{pmatrix}
		\cos{\varphi}&-\sin{\varphi} \\
		\sin{\varphi} &\phantom{-}\cos{\varphi}
	\end{pmatrix}, \quad
	\mathbf{e}_{\varphi}=\begin{pmatrix}
		-\sin{\varphi} \\
		\phantom{-}\cos{\varphi}
	\end{pmatrix}, \quad
	\mathbf{e}^\perp_{\varphi}=-\begin{pmatrix}
		\cos{\varphi} \\
		\sin{\varphi}
	\end{pmatrix}.
\end{eqnarray*}
Further, we define a weighted Radon transform
\begin{equation}
\widetilde{\mathcal{R}}c\left(\mathbf{e}_{\varphi}, s\right):=\int_{B_R} c\brackets{\mathbf{r}} \delta\brackets{\mathbf{r}\cdot\mathbf{e}_\varphi-s} \mathbf{r}
\cdot \mathbf{e}^\perp_{\varphi} \dmathInt \mathbf{r}.
\label{Def:WeightedRadon}
\end{equation}
The next theorem states a link between the MPI forward operator~\eqref{Def:MPIForwardOp}, the Radon transform~\eqref{Def:RadonTransform}, and the weighted Radon transform~\eqref{Def:WeightedRadon}. 
\begin{theorem} \label{Thm:FourierSlice_Sim}
	%\changed{\sout{Suppose that $\abs{s_t=\frac{A}{G}\Lambda\brackets{t}} \leq R $}}.
	Given spatially homogeneous receive coil sensitivities, simultaneous line rotation, and ideal magnetic fields, the MPI forward operator with respect to the l-th receive coil can be written as
	\begin{equation*}
	\mathcal{A}^{\text{FFL}}_l 
	= \mathcal{K}_{1,l} \circ \mathcal{R} +\mathcal{K}_{2,l} \circ \widetilde{\mathcal{R}} + \mathcal{K}_{3,l} \circ \mathcal{R}
	\end{equation*}
	with convolution operators $\mathcal{K}_{i,l}: \, L_2\brackets{Z, \R}	\to L_2\brackets{\R^+, \R}$ for $i=1,2,3$ and $l\in \left\lbrace 1,\dots,L\right\rbrace$
	\begin{eqnarray*}
		\mathcal{K}_{1,l} f\brackets{t}
		&=&-\mu_0 \; A \Lambda'(t) \; \mathbf{e}_{\varphi_t} \cdot \mathbf{p}_l \; \overline{m}' \left(G \; \cdot\right) * f\brackets{\mathbf{e}_{\varphi_t},\cdot}  \left(s_t\right), \\
		\mathcal{K}_{2,l} f\brackets{t}
		&=& \phantom{-} \mu_0  \; G\varphi'_t \; \mathbf{e}_{\varphi_t} \cdot \mathbf{p}_l \; \overline{m}' \left(G \; \cdot\right) * f\brackets{\mathbf{e}_{\varphi_t},\cdot}  \left(s_t\right), \\
		\mathcal{K}_{3,l} f\brackets{t}
		&=& -\mu_0 \; \varphi'_t \; \mathbf{e}^\perp_{\varphi_t}\cdot \mathbf{p}_l \; \overline{m} \left(G \; \cdot\right) 
		* f\brackets{\mathbf{e}_{\varphi_t},\cdot}  \left(s_t\right).
	\end{eqnarray*}
\end{theorem}
\begin{proof}
	Computing the derivative in~\eqref{Def:MPIForwardOp} and assuming spatially homogeneous receive coil sensitivities we obtain
	\begin{eqnarray*}
		\mathcal{A}^{\text{FFL}}_l c\brackets{t} 
		= &-&\mu_0 \; A\Lambda'(t)\; \mathbf{e}_{\varphi_t} \cdot \mathbf{p}_l \;   \int_{B_R} c(\mathbf{r}) \overline{m}' \left( -G \; \mathbf{r} \cdot \mathbf{e}_{\varphi_t} + A\Lambda(t)\right) \dmathInt\mathbf{r} \\
		&+&\mu_0  \; G\varphi'_t \; \mathbf{e}_{\varphi_t} \cdot \mathbf{p}_l \;  \int_{B_R} c(\mathbf{r}) \overline{m}' \left( -G \; \mathbf{r} \cdot \mathbf{e}_{\varphi_t} + A\Lambda(t)\right) \;\mathbf{r}
		\cdot \mathbf{e}^\perp_{\varphi_t}   \dmathInt\mathbf{r} \\
		&-&\mu_0 \; \varphi'_t \; \mathbf{e}^\perp_{\varphi_t}\cdot \mathbf{p}_l \int_{B_R}\; c(\mathbf{r}) \; \overline{m}\left( -G \; \mathbf{r} \cdot \mathbf{e}_{\varphi_t}	+ A\Lambda(t)\right) \dmathInt\mathbf{r}\,,
	\end{eqnarray*}
	where we used $(\mathbf{e}_{\varphi_t})'=\varphi'_t\mathbf{e}^\perp_{\varphi_t}$.
	We proceed similar to~\cite{knopp2011fourier} and rotate the coordinate system $\mathbf{r}':= \mathbf{R}^{-\varphi_t} \; \mathbf{r}$ such that the x-axis gets parallel to the FFL. Using
	\[\mathbf{R}^{\varphi_t}\mathbf{r}' = \mathbf{R}^{\varphi_t}  \begin{pmatrix}v' \\s'\end{pmatrix} = s'\,\mathbf{e}_{\varphi_t} - v'\,\mathbf{e}_{\varphi_t}^\perp\] 
	yields
	\begin{eqnarray*}
		\text{FFL}\brackets{\mathbf{e_{\varphi_t}}, s'} &=& \left\lbrace \mathbf{r} \in \R^2 : \; \mathbf{r} \cdot \mathbf{e}_{\varphi_t} = s'  \right\rbrace = 
		\left\lbrace s'\,\mathbf{e}_{\varphi_t} - v'\,\mathbf{e}_{\varphi_t}^\perp : \; v' \in \R\right\rbrace \\
		&=& \left\lbrace \mathbf{R}^{\varphi_t}  \begin{pmatrix}
			v' \\
			s'
		\end{pmatrix} : \; v' \in \R\right\rbrace
	\end{eqnarray*}
	and with $\mathbf{R}^{-\varphi_t}B_R=B_R$ we obtain 
	\begin{eqnarray*}
		\mathcal{A}^{\text{FFL}}_l c\brackets{t} 
		= &-&\mu_0 \; A\Lambda'(t) \; \mathbf{e}_{\varphi_t} \cdot \mathbf{p}_l    \int_{\R} \mathcal{R}c\left( \mathbf{e}_{\varphi_t}, s'\right)  \overline{m}' \left( -G \; s' + A\Lambda(t)\right) \dmathInt s' \\
		&+&\mu_0  \; G\varphi'_t \; \mathbf{e}_{\varphi_t} \cdot \mathbf{p}_l  \int_{\R}  \widetilde{\mathcal{R}}c\left(\mathbf{e}_{\varphi_t}, s'\right)  \overline{m}' \left( -G \; s' + A\Lambda(t)\right) \dmathInt s' \\
		&-&\mu_0  \; \varphi'_t \; \mathbf{e}^\perp_{\varphi_t}\cdot \mathbf{p}_l \int_{\R} \; \mathcal{R}c\left(\mathbf{e}_{\varphi_t}, s'\right) \; \overline{m}\left( -G \; s' + A\Lambda(t)\right) \dmathInt s'.
	\end{eqnarray*}
	%		In the second integral we cannot extract the Radon transform of the particle concentration itself but we rather a weighted version it
	%		\begin{equation*}
	%		\widetilde{\mathcal{R}}c\left( \mathbf{e}_{\varphi_t}, s'\right) := \int_{\mathbb{R}} c\left(s'\;\mathbf{e}_{\varphi_t} + v'\;\mathbf{e}_{\varphi_t}^\perp\right) v' \dmathInt v' = - \int_{\mathbb{R}^2} c\brackets{\mathbf{r}} \delta\brackets{\mathbf{r}\cdot\mathbf{e}_{\varphi_t}-s'} \mathbf{r} \cdot \mathbf{e}_{\varphi_t}^\perp \dmathInt \mathbf{r}.
	%		\end{equation*}
	Hence, we finally get
	\begin{eqnarray*}
		\mathcal{A}^{\text{FFL}}_l 
		= \mathcal{K}_{1,l} \circ \mathcal{R} +\mathcal{K}_{2,l} \circ \widetilde{\mathcal{R}} + \mathcal{K}_{3,l} \circ \mathcal{R}
	\end{eqnarray*}
	completing the proof.
\end{proof}
\begin{remark}
	The operators $\mathcal{K}_{2,l}$ and $\mathcal{K}_{3,l}$ result from the extra derivatives,  due to the additional time dependencies in~\eqref{Def:MPIForwardOp} for simultaneous line rotation. For sequential line rotation these terms vanish as we choose $\varphi_t$ to be piecewise constant and Theorem~\ref{Thm:FourierSlice_Sim} reduces to Theorem~\ref{Thm:FourierSlice}.
\end{remark}
\begin{remark}
	Regarding an oscillating non-rotating FFL with a phantom continuously rotating in opposite direction to the FFL rotation in the simultaneous setting, we would obtain
	\begin{eqnarray*}
		\mathcal{A}^{\text{FFL}}_l 
		= \mathcal{K}_{1,l} \circ \mathcal{R} +\mathcal{K}_{2,l} \circ \widetilde{\mathcal{R}} 
	\end{eqnarray*}	
	as forward operator. For simultaneous line rotation the term $\mathcal{K}_{3,l}$ results from the temporal change of the magnetic field orientation with respect to the receive coils sensitivity $\mathbf{p}_l$ and thus, is not present in the just considered case.
\end{remark}
Since the two additional terms scale with the FFL rotation speed, we assume, as proposed in~(\cite{bringout2020new},~\cite{knopp2011fourier}), that Theorem~\ref{Thm:FourierSlice} can still be used in case $\varphi_t^\prime$ is suitably small. To further confirm this, we conclude this section by deriving upper bounds for $\abs{\mathcal{K}_{2,l} \widetilde{\mathcal{R}}c\brackets{t}}$ respectively $\abs{\mathcal{K}_{3,l} \mathcal{R}c\brackets{t}}.$
\begin{lemma} \label{Lemma:Int2}
	Let $R=\frac{A}{G}$, i.e. we suppose that the particle concentration is completely located within the sampling region as usually $\frac{A}{G}$ is the maxium displacement of the FFL. Then, we have that
	\begin{eqnarray*}
		\abs{\mathcal{K}_{1,l} \mathcal{R}c\brackets{t}}
		&\geq&\frac{\abs{\Lambda'\brackets{t}}}{\varphi_t'}\abs{\mathcal{K}_{2,l}\widetilde{\mathcal{R}}c\brackets{t}}.
	\end{eqnarray*}
\end{lemma}
\begin{proof}
	It holds for $\brackets{\mathbf{e}_\varphi,s} \in Z$ that
	\begin{eqnarray*}
		\widetilde{\mathcal{R}}c\left(\mathbf{e}_{\varphi}, s\right)&=&\int_{B_R} c\brackets{\mathbf{r}} \delta\brackets{\mathbf{r}\cdot\mathbf{e}_\varphi-s} \mathbf{r}
		\cdot \mathbf{e}^\perp_{\varphi} \dmathInt \mathbf{r}
		\\
		&=&-\int_{-\sqrt{R^2-s^2}}^{\sqrt{R^2-s^2}} c\left(s\;\mathbf{e}_{\varphi} - v\;\mathbf{e}_{\varphi}^\perp\right) v \dmathInt v 
		\leq
		R \; \mathcal{R}c\left( \mathbf{e}_{\varphi}, s\right).
	\end{eqnarray*}
	Hence, we get
	\begin{eqnarray*}
		\abs{\mathcal{K}_{1,l} \mathcal{R}c\brackets{t}}
		&=&\abs{\mu_0 \; A \Lambda'(t) \; \mathbf{e}_{\varphi_t} \cdot \mathbf{p}_l \; \overline{m}' \left(G \; \cdot\right) * \mathcal{R}c\brackets{\mathbf{e}_{\varphi_t},\cdot}  \left(s_t\right)}	\\	
		&=&\abs{\mu_0 \; G R \Lambda'(t) \; \mathbf{e}_{\varphi_t} \cdot \mathbf{p}_l \; \overline{m}' \left(G \; \cdot\right) * \mathcal{R}c\brackets{\mathbf{e}_{\varphi_t},\cdot}  \left(s_t\right)} \\
		&\geq&\abs{\mu_0 \; G \Lambda'(t) \; \mathbf{e}_{\varphi_t} \cdot \mathbf{p}_l \; \overline{m}' \left(G \; \cdot\right) * \widetilde{\mathcal{R}}c\brackets{\mathbf{e}_{\varphi_t},\cdot}  \left(s_t\right)}
		\\
		&=&\abs{\frac{\Lambda'\brackets{t}}{\varphi_t'}\mu_0 \; G \varphi_t' \; \mathbf{e}_{\varphi_t} \cdot \mathbf{p}_l \; \overline{m}' \left(G \; \cdot\right) * \widetilde{\mathcal{R}}c\brackets{\mathbf{e}_{\varphi_t},\cdot}  \left(s_t\right)}\\
		&=&\frac{\abs{\Lambda'\brackets{t}}}{\varphi_t'}\abs{\mathcal{K}_{2,l}\widetilde{\mathcal{R}}c\brackets{t}}.
	\end{eqnarray*}
\end{proof}
\\
The ratio in the last lemma relates the FFL translation with the rotation speed and thus, if it is sufficiently large, $\mathcal{K}_{2,l}\widetilde{\mathcal{R}}c\brackets{t}$ might be neglected in the image reconstruction process. As already mentioned, usually the excitation function is chosen to be sinusoidal, e.g. $\Lambda\brackets{t}=-\cos\brackets{2\pi f_{\text{d}}t}$, with drive frequency $f_{\text{d}}>0$. Then, we get that
\begin{equation}
\frac{\abs{\Lambda'\brackets{t}}}{\varphi_t'} = \frac{f_{\text{d}}}{f_{\text{rot}}}\abs{\sin\brackets{2\pi f_{\text{d}}t}}.
\label{Eqn:SpeedRatio}
\end{equation}
At each turning point of the FFL the first term $\mathcal{K}_{1,l}\mathcal{R}c$ becomes zero. By~\eqref{Eqn:SpeedRatio} the time intervals, in which the second term $\mathcal{K}_{2,l}\widetilde{\mathcal{R}}c$ might overweight the first integral, are the smaller the larger the ratio $\frac{f_{\text{d}}}{f_{\text{rot}}}$ gets. According to~\cite{weizenecker2018fokker} drive frequencies are around $1$ kHz to $150$ kHz. A typical upper bound for the rotation frequency is $100$~Hz \cite{bringout2020new} leading to $\frac{f_{\text{d}}}{f_{\text{rot}}} \geq 10.$ At last, we state an example of particle concentrations leading to a vanishing $\mathcal{K}_{2,l}\widetilde{\mathcal{R}}c$.
\begin{example}
	Let $c$ be radial symmetric, i.e. $c(\mathbf{r})=c(\left| \mathbf{r}\right| )$. We compute
	\begin{eqnarray*}
		-\widetilde{\mathcal{R}}c\left(\mathbf{e}_{\varphi}, s\right)&=&\int_{-\sqrt{R^2-s^2}}^{\sqrt{R^2-s^2}} c\left(s\;\mathbf{e}_{\varphi} - v\;\mathbf{e}_{\varphi}^\perp\right) v \dmathInt v 
		= \int_{-\sqrt{R^2-s^2}}^{\sqrt{R^2-s^2}} c\left(\sqrt{s^2 + v^2} \right) \; v \dmathInt v = 0.
	\end{eqnarray*}
\end{example}
Next, we give an estimate for the third integral $\mathcal{K}_{3,l} \mathcal{R}c\brackets{t}.$
\begin{lemma}\label{Lemma:Int3}
	With $m$ denoting the magnetic moment of a single particle and $N_p$ the total amount of particles contained in the tracer injection, it holds that
	\begin{eqnarray*}
		\abs{\mathcal{K}_{3,l} \mathcal{R}c\brackets{t}}
		&\leq&	\mu_0 \; \varphi'_t \; \left\|  \mathbf{p}_l \right\|  \; m \; N_p.
	\end{eqnarray*}
\end{lemma}
\begin{proof}
	According to the proof of Theorem~\ref{Thm:FourierSlice_Sim} we have
	\begin{eqnarray*}
		\abs{\mathcal{K}_{3,l} \mathcal{R}c\brackets{t}} 
		&=& \abs{ \mu_0 \; \varphi'_t \; \mathbf{e}^\perp_{\varphi_t}\cdot \mathbf{p}_l \int_{B_R}\; c(\mathbf{r}) \; \overline{m}\left( -G \; \mathbf{r} \cdot \mathbf{e}_{\varphi_t}	+ A\Lambda(t)\right) \dmathInt\mathbf{r}}.
	\end{eqnarray*}
	The modulus of the mean magnetic moment is bounded by $m$ (cf. Langevin model in Section~\ref{Sec:FFL}) yielding
	\begin{eqnarray*}
		\abs{\mathcal{K}_{3,l} \mathcal{R}c\brackets{t}}
		&\leq&
		\abs{\mu_0 \; \varphi'_t \; \mathbf{e}^\perp_{\varphi_t}\cdot \mathbf{p}_l \; m   
			\int_{B_R}\; c(\mathbf{r}) \dmathInt \mathbf{r}} \leq 
		\mu_0 \; \varphi'_t \; \left\|  \mathbf{p}_l \right\|  \; m \; N_p.
	\end{eqnarray*}
\end{proof}

\begin{remark}
	For practical reasons it might be convenient to state an upper bound in terms of the maximal particle concentration $c_{\max}$. This can be easily obtained by using $N_p \leq c_{\max}\pi R^2.$ 
\end{remark}
In case the phantom is fully located within the saturation area on one side of the FFL, the above inequality holds approximately with equality. The estimate in the last Lemma can be computed beforehand to measurements as all components are determined by the scanner setup and the choice of the injected tracer. Thus, this upper bound can be determined and compared to the magnitudes of measured data in order to evaluate whether incorporation is needed. From the $s$-shape of the Langevin function it follows that $\abs{\mathcal{K}_{3,l} \mathcal{R}c\brackets{t}}$ is largest at turning points of the FFL, which are the zero crossings of $\abs{\mathcal{K}_{1,l} \mathcal{R}c\brackets{t}}$.

%%%%%%%%%%%%%%%%%%%%%%%%%%%%%%%%%%
\section{TV regularized Image Reconstruction} \label{ImageReco}
%%%%%%%%%%%%%%%%%%%%%%%%%%%%%%%%%%
Inspired by~\cite{Tovey_2019} we reconstruct particle concentration and Radon data simultaneously via total variation regularization. To this end, we first introduce the space of functions of bounded variation $\BV$  on the domain $B_R$ 
\begin{eqnarray*}
	\BV\brackets{B_R} := \left\lbrace c \in L_1\brackets{B_R,\R} : \; \TV\brackets{c} < \infty\right\rbrace ,  
\end{eqnarray*}
with
\begin{equation*}
\TV\brackets{c} := \sup\left\lbrace \int_{\R^2} c\brackets{\mathbf{r}}  \operatorname{div}\brackets{\mathbf{g}}\brackets{\mathbf{r}} \dmathInt \mathbf{r} : \; \mathbf{g} \in C^\infty_0\brackets{B_R,\R^2}, \; \abs{\mathbf{g}\brackets{\mathbf{r}}}_2 <1 \; \text{for all} \; \mathbf{r} \right\rbrace.
\end{equation*}
Equipped with the norm $\norm{\cdot}_{\BV} := \norm{\cdot}_{L_1} + \TV\brackets{\cdot}$ this space becomes a Banach space and it holds $\BV\brackets{B_R} \subset L_2\brackets{B_R,\R}$. Additionally, the Poincaré-Wirtinger inequality holds
\begin{equation}
\norm{c - \overline{c}}_{L_2} \leq C\, \TV\brackets{c}, \quad \overline{c} = \frac{1}{\pi R^2} \int_{B_R} c\brackets{\mathbf{r}} \dmathInt \mathbf{r}
\label{Eqn_Poincare}
\end{equation}
for some constant $C>0$.
For more information concerning TV we recommend to consult e.g.~\cite{acar1994analysis} and \cite{burger2013guide}.\\\\
Define $\mathcal{D} := L_2\brackets{B_R,\R}\times L_2\brackets{Z,\R }$ and let $A_l: \mathcal{D}   \to L_2\brackets{\R^+,\R}$ such that
\begin{equation}
A_l\brackets{c,v} := \mathcal{K}_{1,l} v + \brackets{\mathcal{K}_{2,l} \circ \widetilde{\mathcal{R}}}c + \mathcal{K}_{3,l} v
\label{Eqn:NewOp}
\end{equation}
with operators $\mathcal{K}_{i,l},\; i=1,2,3$ defined in Theorem~\ref{Thm:FourierSlice_Sim}. It holds that $A_l\brackets{c,\mathcal{R}c}=\mathcal{A}^{\text{FFL}}_l c$ for $c\geq0$. Now, we are able to state the minimization problem we want to solve
\begin{equation}
\min_{\brackets{c,v} \in \mathcal{C}} \frac{1}{2}\sum_l \norm{A_l\brackets{c,v} - u_l}^2_{L_2} + \frac{\alpha_1}{2}\norm{\mathcal{R}c - v}^2_{L_2} + \alpha_2 \TV\brackets{c},
\label{Eqn:OptProblem}
\end{equation}
with feasible set $\mathcal{C} := \left\lbrace \brackets{c,v} \in \mathcal{D} : \; c\geq 0,\; v\geq0\right\rbrace,$ given data $u_l,\; l=1,\dots,L$, weighting parameter $\alpha_1$, and regularization parameter $\alpha_2>0$. Note that actively we only penalize the choice of the particle concentration. Nevertheless, if needed an additional regularization term acting on the Radon data $v$ can be included, e.g. directional TV regularization as used in~\cite{Tovey_2019}.
\begin{theorem}
	The minimization problem~\eqref{Eqn:OptProblem} has a solution $\brackets{c^*,v^*} \in \mathcal{C}$ and $c^* \in \BV\brackets{B_R}.$
\end{theorem}
\begin{proof}
	Rewriting the constrained optimization problem~\eqref{Eqn:OptProblem} using the indicator function $\delta_{\mathcal{C}}$ yields
	\begin{equation*}
	\min_{\brackets{c,v} \in \mathcal{D}}  J\brackets{c,v},
	\end{equation*}
	\begin{equation*}
	J\brackets{c,v}:=\frac{1}{2}\sum_l \norm{A_l\brackets{c,v} - u_l}^2_{L_2} + \frac{\alpha_1}{2}\norm{\mathcal{R}c - v}^2_{L_2} + \alpha_2 \TV\brackets{c} + \delta_{\mathcal{C}}\brackets{c,v}.
	\end{equation*}
	Note that $\mathcal{D}$ is a Hilbertspace with inner product 
	\begin{equation*}
	\left\langle \cdot, \cdot \right\rangle_{\mathcal{D}}:=\left\langle \cdot, \cdot \right\rangle_{L_2\brackets{B_R,\R} }+\left\langle \cdot, \cdot \right\rangle_{L_2\brackets{Z,\R }}.
	\end{equation*}
	Obviously $J$ is proper. The total variation $\TV\brackets{c}$ is convex and weakly lower semicontinuous according to~\cite{acar1994analysis}. Further, $A_l$ and $\mathcal{R}$ are linear bounded operators and the feasible set $\mathcal{C}$ is closed and convex. Thus, $J$ is jointly convex in $\brackets{c,v}$ and weakly lower semicontinuous. For existence the only remaining part to show is that $J$ is also coercive. From~\cite{acar1994analysis} we know that
	\begin{equation}
	\norm{c}_{\BV} \to \infty \quad \Longrightarrow \quad \frac{\alpha_1}{2}\norm{\mathcal{R}c - v}^2_{L_2} + \alpha_2 \TV\brackets{c} \to \infty
	\label{Eqn:BV_Coercivity}
	\end{equation}
	as $\mathcal{R}c$ does not annihilate for constant functions. Regard now
	\begin{equation*}
	\norm{\brackets{c,v}}_{\mathcal{D}}^2 = \norm{c}_{L_2}^2 + \norm{v}_{L_2}^2 \to \infty.
	\end{equation*}
	Then, either $\norm{c}_{L_2}$ or $\norm{v}_{L_2}$ needs to tend to infinity. If $\norm{v}_{L_2} \to \infty$ we directly get that $J\brackets{c,v}\to \infty$ as $\norm{\mathcal{R}c - v}^2_{L_2} \to \infty$. Thus, we assume that $\norm{v}_{L_2} $ is bounded and $\norm{c}_{L_2} \to \infty$. From~\eqref{Eqn_Poincare} we get with $\mathbf{1}_{B_R}$ denoting the characteristic function
	\begin{eqnarray*}
		\norm{c}_{L_2} &\leq& \norm{c-\overline{c}\,\mathbf{1}_{B_R}}_{L_2} + \norm{\overline{c}\,\mathbf{1}_{B_R}}_{L_2} \leq C \,\TV\brackets{c} + \sqrt{\pi R^2}\abs{\overline{c}} \\
		&=& C \,\TV\brackets{c} + \frac{1}{\sqrt{\pi R^2}}\norm{c}_{L_1} 
		\leq \max\left\lbrace C, \frac{1}{\sqrt{\pi R^2}}\right\rbrace \norm{c}_{\BV} 
	\end{eqnarray*}	
	and thus, $\norm{c}_{\BV} \to \infty$ if $\norm{c}_{L_2} \to \infty$. Together with~\eqref{Eqn:BV_Coercivity} we finally obtain that $J$ is coercive yielding existence of a minimizer. 
\end{proof}

%%%%%%%%%%%%%%%%%%%%%%%%%%%%%%%%%%
\section{Numerical Results} \label{Results}
%%%%%%%%%%%%%%%%%%%%%%%%%%%%%%%%%%
Finally, we state numerical results for synthetic data. Thereby, the basis for our data simulation is given by the framework developed by Gael Bringout~\cite{bringout2016field} available at~\url{https://github.com/gBringout}. Accordingly, the tracer is modeled as a solution with $0.5 \, \frac{\text{mol}\brackets{\text{Fe}_3\text{O}_4}}{\text{m}^3}$ concentration of magnetite with $30$ nm core diameter and $\frac{0.6}{\mu_0}$ T saturation magnetization. As excitation function we choose
\begin{equation}
\Lambda\brackets{t} = - \cos \brackets{2\pi f_{\text{d}} t}
\label{Eqn:Excitation}
\end{equation}
with drive frequency $f_{\text{d}}$. Further parameters can be found in Table~\ref{table:SimulationParameters}. Our concentration phantom (cf. Figure~\ref{Fig:Phantom}) is normalized to one and located within a circle around the origin with radius $\frac{A}{G}$, which is the maximum displacement of the FFL. Thus, we set the FOV to be $\left[-\frac{A}{G},\frac{A}{G} \right] \times \left[-\frac{A}{G},\frac{A}{G} \right]$. For data generation we divide the FOV into $501 \times 501$, for image reconstruction into $201\times 201$ pixel. In case of sequential line rotation, we gather data for $25$ sweeps of the FFL through the FOV and angles equally distributed in $\left[ 0,\pi\right]$. More precisely, we regard angles $\varphi_j = \brackets{j-1}\frac{\pi}{25},\; j=1,\dots,25$. Considering simultaneous line rotation, we choose a total measurement time of $\frac{1}{2f_{\text{rot}}}$ resulting in the same amount of FFL translations through the phantom covering angles again in $\left[ 0,\pi\right]$. 
\begin{table}[htbp]
	\centering
	\caption{Simulation parameters }
	\renewcommand{\arraystretch}{1.3}
	\begin{tabular}{|c|c|c|c|}
		\hline
		\textbf{Parameter} & \textbf{Explanation} & \textbf{Value} & \textbf{Unit} \\
		\hline \hline
		& & & \\[-1.2em]
		$\mu_0$ & magnetic permeability & $4\pi\cdot 10^{-7}$ & $\text{T}\text{m}\text{A}^{-1}$ \\
		$k_{\text{B}}$ & Boltzmann constant & $1.380650424\cdot 10^{-23}$ & $\text{J}\text{K}^{-1}$ \\
		\hline
		$G$ & gradient strength & $4$ & $\text{T}\brackets{\text{m}\mu_0}^{-1}$ \\
		$A$ & drive peak amplitude & $15$  & $\text{mT}\mu_0^{-1}$\\
		$\mathbf{p}_1$ & sensitivity of the first receive coil & $\left[0.015/293.29, 0 \right]^T $  & \\
		$\mathbf{p}_2$ & sensitivity of the second receive coil & $\left[0, 0.015/379.71\right]^T $  & \\
		\hline
		$f_{\text{d}}$ & drive-field frequency & 25 & kHz\\
		$f_{\text{rot}}$ & line rotation frequency & 1 & kHz\\
		$f_{\text{s}}$ & sampling frequency & 8 & MHz\\
		\hline
	\end{tabular}
	\label{table:SimulationParameters}
	\renewcommand{\arraystretch}{1}
\end{table}\\
In the following, we denote
\begin{equation*}
\widehat{\mathcal{K}}_{i,l}f\brackets{t}:=\frac{\mathcal{K}_{i,l}f\brackets{t}}{\underset{t}{\max}\left\lbrace \abs{ \mathcal{K}_{1,l} f\brackets{t}} \right\rbrace }, \quad \text{for}\; i=1,2,3,\;l=1,2.
\end{equation*}
For our reconstructions, we neglect the second term of the forward operator~\eqref{Eqn:NewOp}, which is justified due to Lemma~\ref{Lemma:Int2} and the choice of frequencies. This is emphasized by Figure~\ref{Fig:IntegralEstimates}. The left plot shows $\widehat{\mathcal{K}}_{1,1} \mathcal{R}c$ in comparison to $\widehat{\mathcal{K}}_{3,1} \mathcal{R}c$, whereas the right plot images $\widehat{\mathcal{K}}_{2,1} \widetilde{\mathcal{R}}c + \widehat{\mathcal{K}}_{3,1} \mathcal{R}c$ and $\widehat{\mathcal{K}}_{3,1} \mathcal{R}c$. We find that both additional terms are small compared to $\widehat{\mathcal{K}}_{1,1} \mathcal{R}c$. Especially the second term seems to have only an influence, when the main part reaches its highest values, while the third term is largest for the zero crossings of $\widehat{\mathcal{K}}_{1,1} \mathcal{R}c$. \\
\begin{figure}[htbp]
	\begin{subfigure}[b]{0.5\textwidth}
		\centering
		\includegraphics[width=1.2\linewidth, trim=2.8cm 8cm 0 8cm, clip]{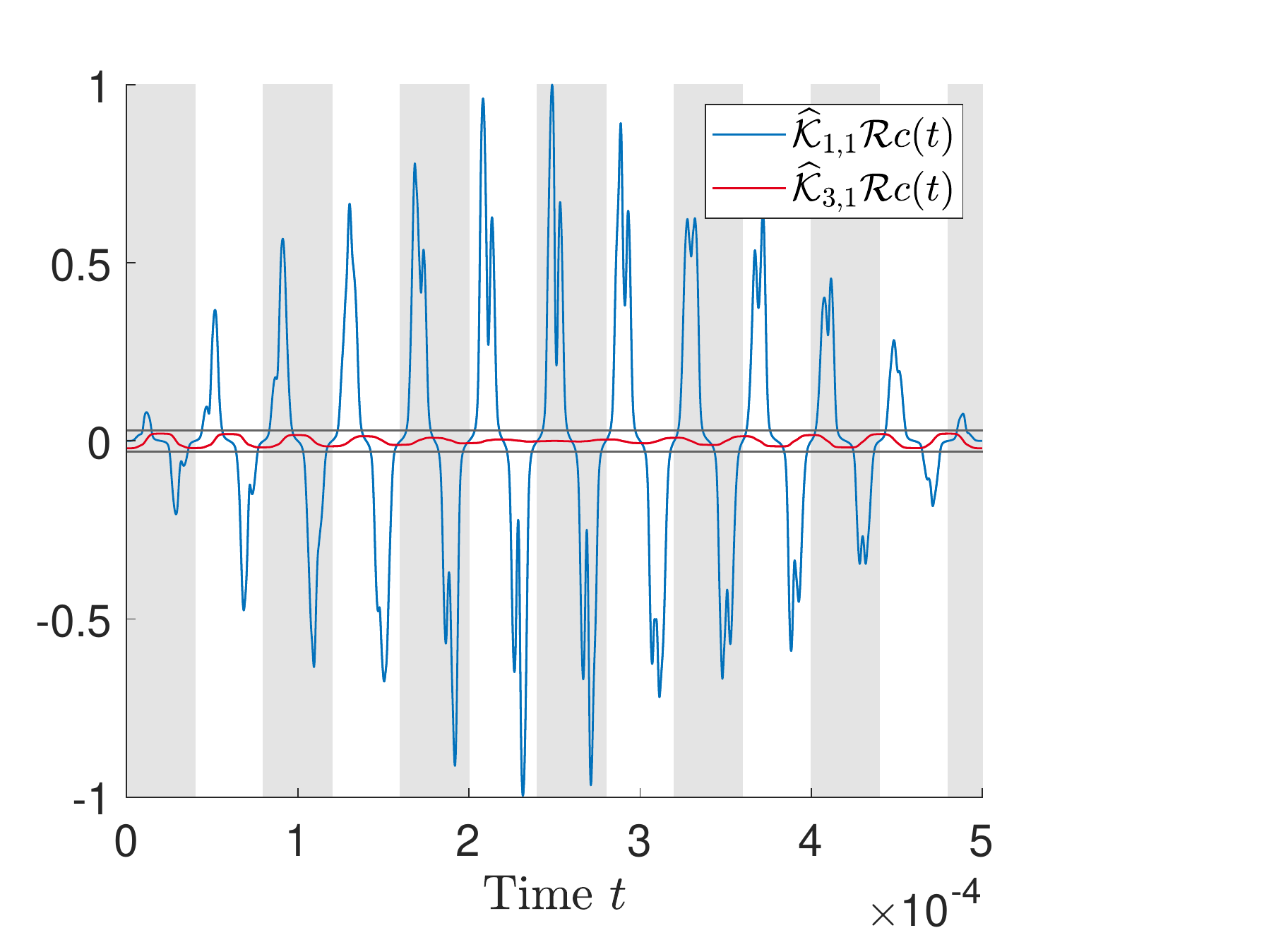}
		\subcaption{}
		\label{Fig:IntegralEstimates_2}
	\end{subfigure}
	\hfill
	\begin{subfigure}[b]{0.5\textwidth}
		\centering
		\includegraphics[width=1.2\linewidth, trim=2.8cm 8cm 0cm 8cm, clip]{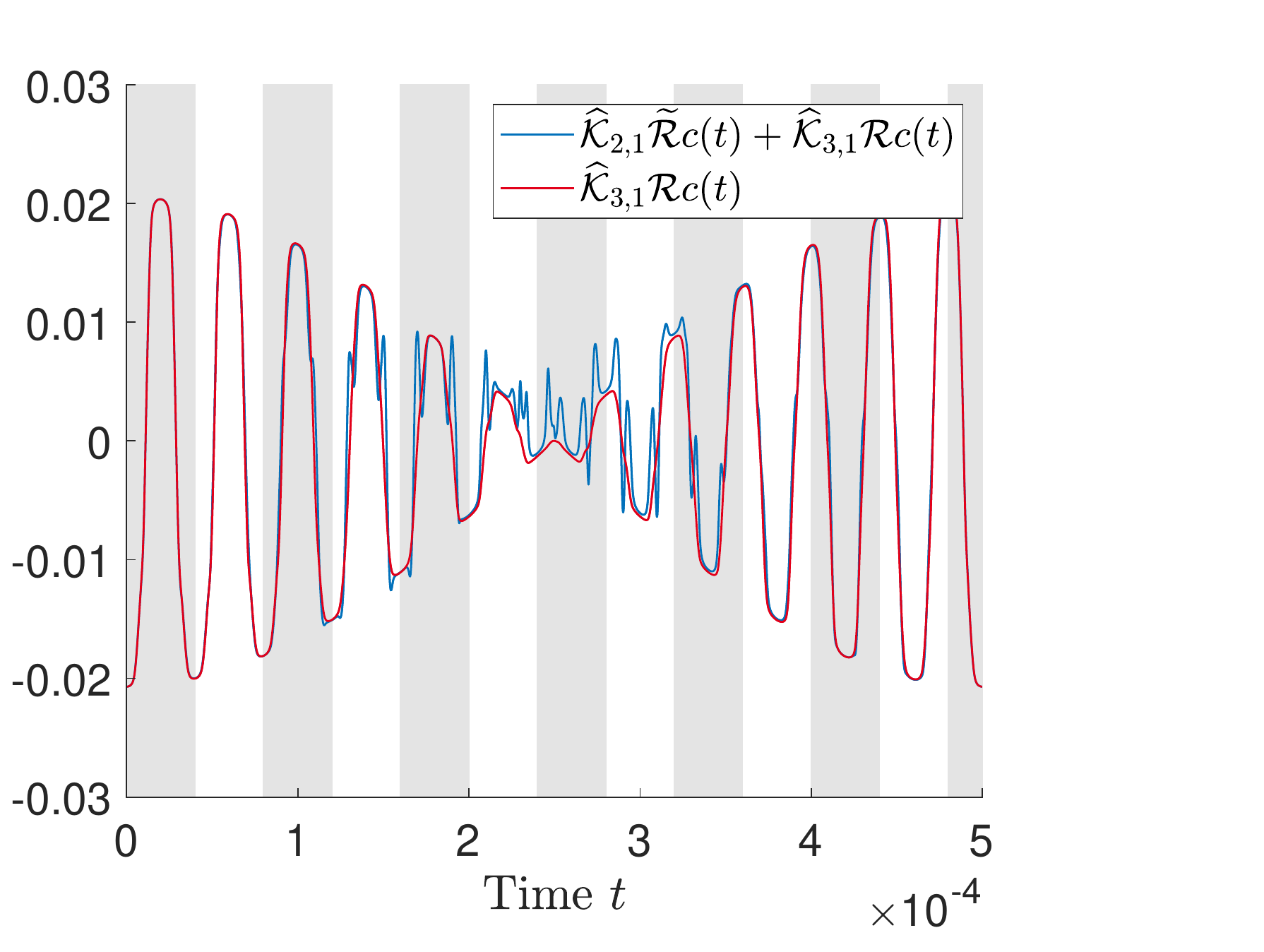}
		\subcaption{}
		\label{Fig:IntegralEstimates_3}
	\end{subfigure}
	\caption{Plots showing $\widehat{\mathcal{K}}_{1,1} \mathcal{R}c$ (blue) in comparison to $\widehat{\mathcal{K}}_{3,1} \mathcal{R}c$ (red) together with the bound determined in Lemma~\ref{Lemma:Int3} (dark grey) (a) and $\widehat{\mathcal{K}}_{2,1} \widetilde{\mathcal{R}}c + \widehat{\mathcal{K}}_{3,1} \mathcal{R}c$ (blue) in comparison to $\widehat{\mathcal{K}}_{3,1} \mathcal{R}c$ (red) (b) for the phantom in Figure~\ref{Fig:Phantom}. Blocks in the background of the two plots demonstrate periods of the drive field. The total measurement time is $ 12.5\frac{1}{f_{\text{d}}}=\frac{1}{2f_{\text{rot}}}$.}
	\label{Fig:IntegralEstimates}
\end{figure}\\
Figure~\ref{Fig:Phantom} shows the phantom which we want to reconstruct. In order to evaluate the discretized version of~\eqref{Eqn:NewOp} for a specific time point, the corresponding column in a sinogram filled angle by angle is needed. Thus, to be able to compute data for every sampling point, a sinogram containing a column for each angle the FFL attains during scanning is needed. For sequential line rotation this amounts to the sinogram shown in Figure~\ref{Fig:Sino_discr}, for simultaneous line rotation to the one in Figure~\ref{Fig:Sino_cont_1}. In order to reduce the problem size, we rather aim at reconstructing only the sinogram taking values along the dashed line in Figure~\ref{Fig:Sino_cont_1} resulting in Figure~\ref{Fig:Sino_cont}. Thereby, we get an additional error but because of the shape of the convolution kernel of the main part $\widehat{\mathcal{K}}_{1,l} \mathcal{R}c$, which converges to the dirac delta for particle diameters tending to infinity~\cite{knopp2012magnetic}, and using a regularization method for reconstruction this should be no further problem. 
\begin{figure}[htbp]
	\begin{subfigure}[b]{0.5\textwidth}
		\centering
		\includegraphics[width=0.9\linewidth, trim=0 8.5cm 0 8cm, clip]{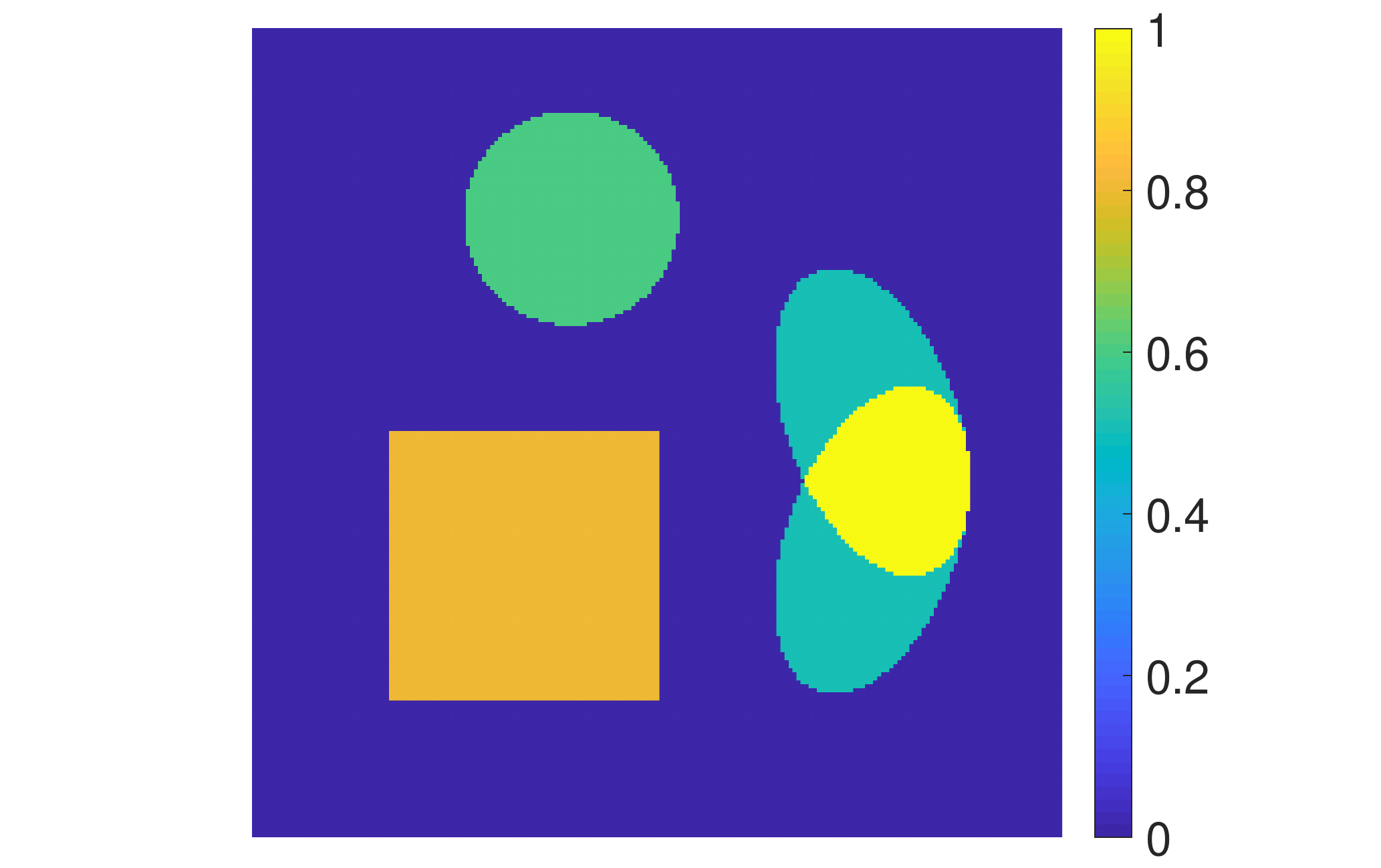}
		\subcaption{}
		\label{Fig:Phantom}
	\end{subfigure}
	\hfill
	\begin{subfigure}[b]{0.5\textwidth}
		\centering
		\includegraphics[width=1.\linewidth, trim=0.5cm 9.2cm 0cm 8cm, clip]{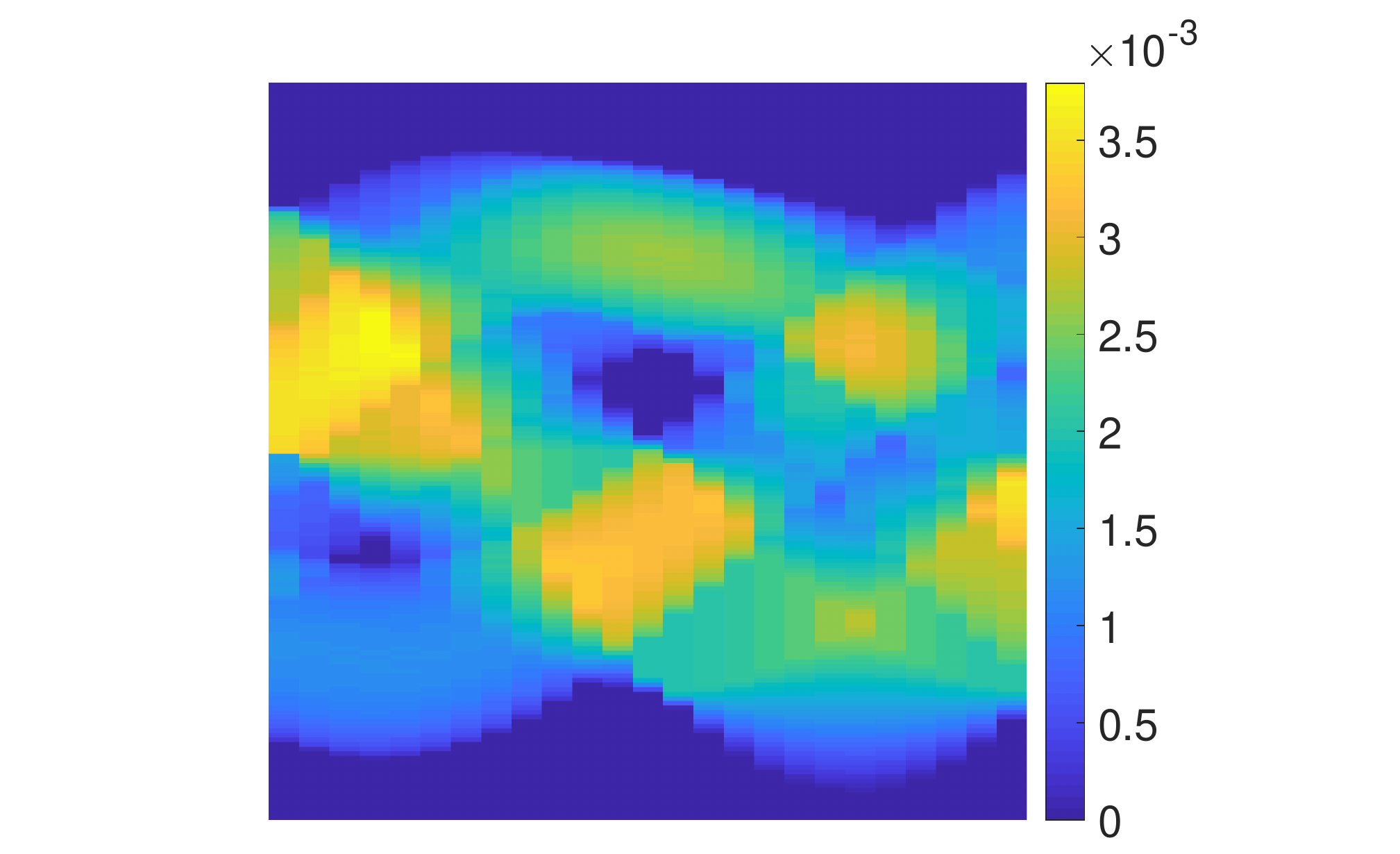}
		\subcaption{}
		\label{Fig:Sino_discr}
	\end{subfigure}
	\caption{Phantom (a) and corresponding sinogram filled angle by angle (b).}
\end{figure}
\begin{figure}[htbp]
	\begin{subfigure}[b]{0.5\textwidth}
		\centering
		\includegraphics[width=0.975\linewidth, trim=3.5cm 11cm 2cm 8cm, clip]{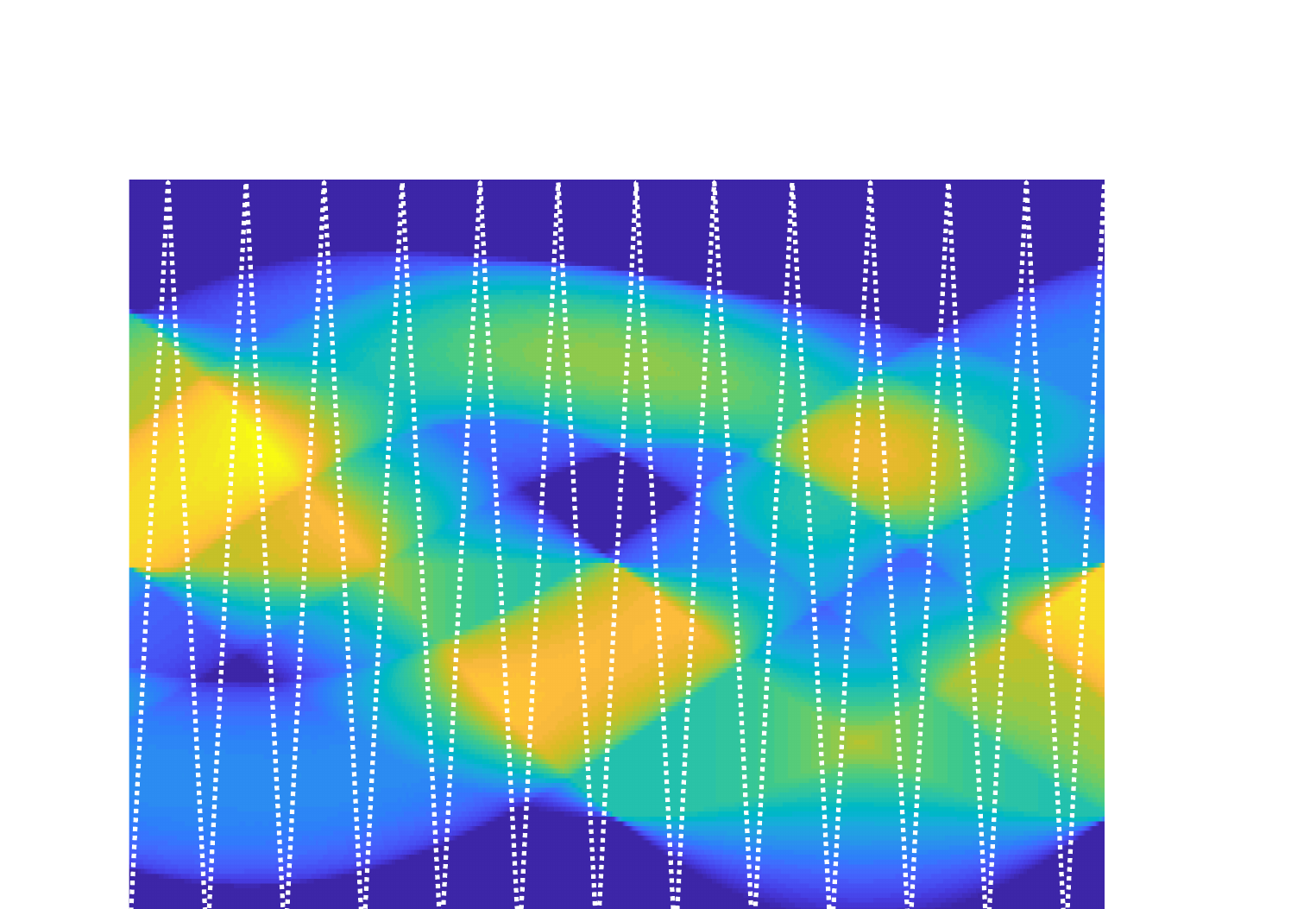}
		\subcaption{}
		\label{Fig:Sino_cont_1}
	\end{subfigure}
	\hfill
	\begin{subfigure}[b]{0.5\textwidth}
		\centering
		\includegraphics[width=0.8\linewidth, trim=3.5cm 10cm 2cm 8cm, clip]{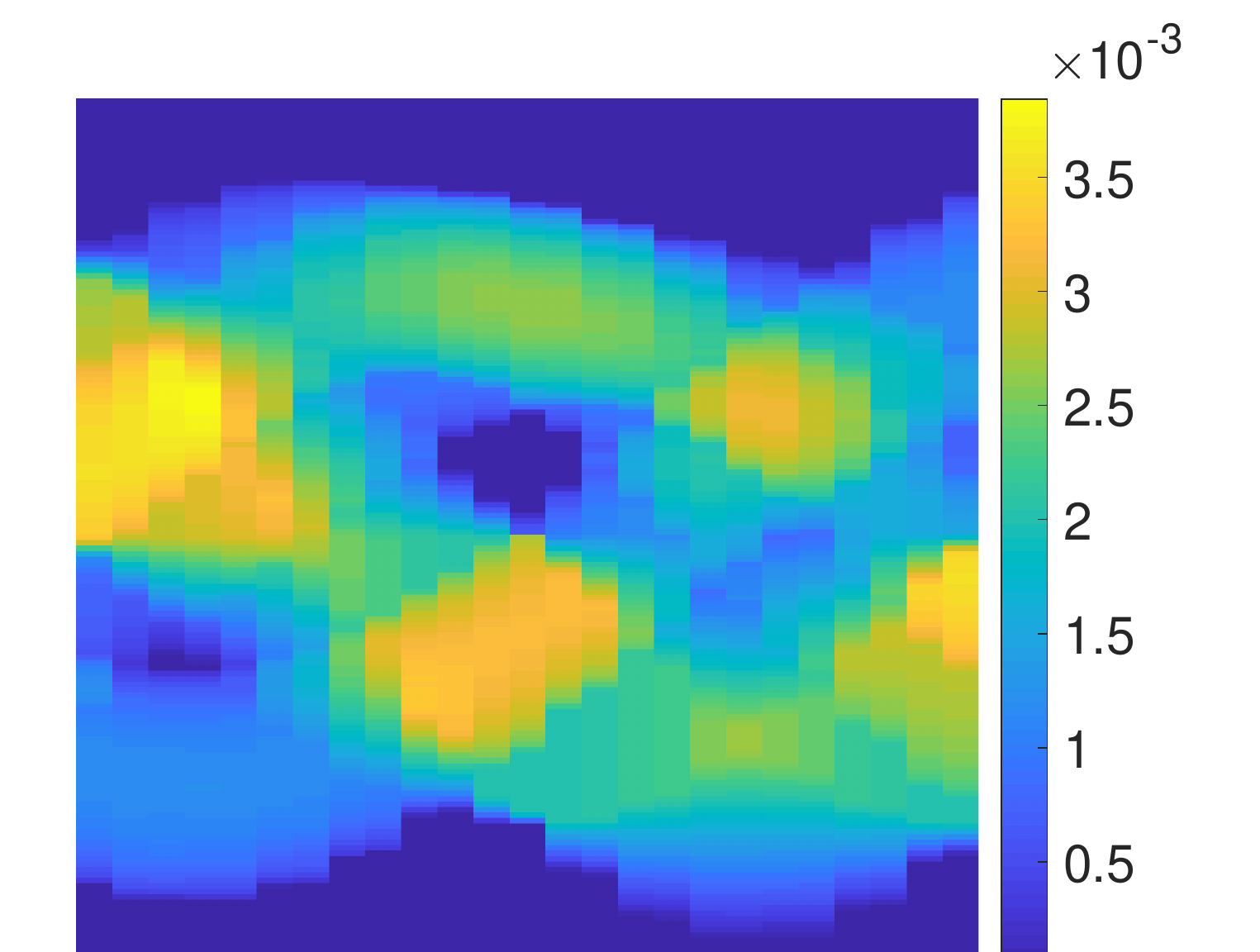}
		\subcaption{}
		\label{Fig:Sino_cont}
	\end{subfigure}
	\caption{(a) Sinogram for the phantom shown in Figure~\ref{Fig:Phantom} with columns for each angle attained by the FFL during scanning. (b) Adapted version filled following the dashed line in (a). }
\end{figure}\\
We discretize~\eqref{Eqn:OptProblem} by applying standard methods. We define discretized versions of the Radon transform $\mathbf{R}_{\text{seq}}$ and $\mathbf{R}_{\text{sim}}$ such that $\mathbf{R}_{\text{seq}}\mathbf{c}$ yields the sinogram in Figure~\ref{Fig:Sino_discr} and $\mathbf{R}_{\text{sim}}\mathbf{c}$ the one in Figure~\ref{Fig:Sino_cont}. Further, we scale data $\mathbf{u}_l$ and discretized forward operators $\mathbf{K}_{i,l}$ by dividing by the maximum absolute data value $u^*$
\begin{equation*}
\widehat{\mathbf{u}}_l := \frac{\mathbf{u}_l}{u^*}, \quad \widehat{\mathbf{K}}_{i,l} := \frac{\mathbf{K}_{i,l}}{u^*}.
\end{equation*}
\begin{table}[htbp]
	\centering
	\caption{Reconstruction methods}
	\renewcommand{\arraystretch}{1.3}
	\begin{tabular}{|c|c|}
		\hline
		\textbf{Method} & \textbf{Regarded minimization problem} \\
		\hline \hline
		& \\[-1.2em]
		$\mathcal{M}_1$ & $ \displaystyle
		\min_{\mathbf{c}\geq 0, \mathbf{v}\geq 0} \frac{1}{2}\sum_{l=1}^2 \norm{\widehat{\mathbf{K}}_{1,l}\mathbf{v} - \widehat{\mathbf{u}}_l}^2_{2} + \frac{\alpha_1}{2}\norm{\mathbf{R}_{\text{seq}}\mathbf{c} - \mathbf{v}}^2_{2} + \alpha_2 \norm{\abs{\nabla \mathbf{c}}_2}_1 $  \\
		\hline 
		$\mathcal{M}_2$ & $ \displaystyle
		\min_{\mathbf{c}\geq 0, \mathbf{v}\geq 0} \frac{1}{2}\sum_{l=1}^2 \norm{\widehat{\mathbf{K}}_{1,l}\mathbf{v} - \widehat{\mathbf{u}}_l}^2_{2} + \frac{\alpha_1}{2}\norm{\mathbf{R}_{\text{sim}}\mathbf{c} - \mathbf{v}}^2_{2} + \alpha_2 \norm{\abs{\nabla \mathbf{c}}_2}_1 $   \\
		\hline
		$\mathcal{M}_3$ & $ \displaystyle
		\min_{\mathbf{c}\geq 0, \mathbf{v}\geq 0} \frac{1}{2}\sum_{l=1}^2 \norm{\brackets{\widehat{\mathbf{K}}_{1,l}+\widehat{\mathbf{K}}_{3,l}}\mathbf{v} - \widehat{\mathbf{u}}_l}^2_{2} + \frac{\alpha_1}{2}\norm{\mathbf{R}_{\text{sim}}\mathbf{c} - \mathbf{v}}^2_{2} + \alpha_2 \norm{\abs{\nabla \mathbf{c}}_2}_1 $   \\
		\hline
	\end{tabular}
	\label{table:RecoMethods}
	\renewcommand{\arraystretch}{1}
\end{table} \\
For reconstruction we regard the versions of~\eqref{Eqn:OptProblem} specified in Table~\ref{table:RecoMethods}. To solve the resulting problem, we use CVX, a package for specifying and solving convex programs~(\cite{cvx},~\cite{gb08}), together with the MOSEK solver~\cite{mosek2010mosek}. For our weighting parameters we choose $\alpha_1 \in \left\lbrace 1 , 2, 4 \right\rbrace\cdot10^4$ and $\alpha_2 \in \left\lbrace 0.1^{5.5-i0.05},\; i=0,\dots,49 \right\rbrace$ such that $\brackets{\alpha_1,\alpha_2}$ maximizes the structural similarity (SSIM) of the reconstructed particle concentration with the groundtruth. However, there might be better parameter choices. Note that $s_t$, i.e. the distance of the FFL to the origin, is not sampled equidistantly due to the choice of the excitation function~\eqref{Eqn:Excitation} and because we sample for equidistant time points~(cf.~\cite{knopp2011fourier}). While in~\cite{knopp2011fourier} they need to process and regrid the signal to be able to use Wiener deconvolution, this is not necessary for our methods.   \\\\
We start with a result for sequential line rotation in Figure~\ref{Fig:Reco_discr_discr}. We obtain good reconstructions for both, the phantom as well as the Radon data. Nevertheless, the vertices of the square are not resolved a hundred percent properly. This is likely due to missing information in the data and because we used an isotropic version of the TV penalty term, which favors rounded edges (cf.~\cite{burger2013guide}). 
\begin{figure}[htbp]
	\begin{subfigure}[b]{0.5\textwidth}
		\centering
		\includegraphics[width=1\linewidth, trim=0 9cm 0 8cm, clip]{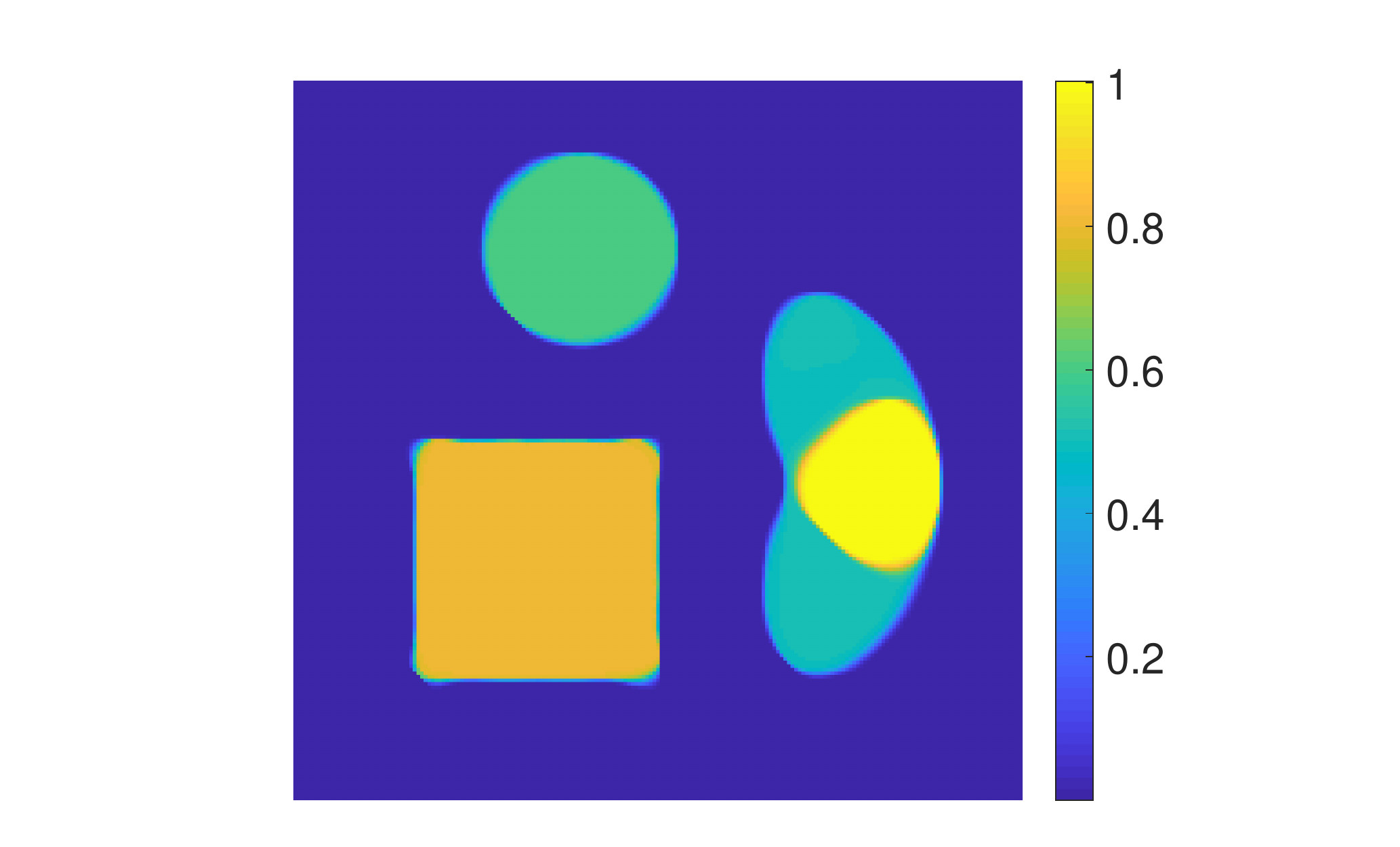}
	\end{subfigure}
	\hfill
	\begin{subfigure}[b]{0.5\textwidth}
		\centering
		\includegraphics[width=1\linewidth, trim=0cm 9cm 0cm 8cm, clip]{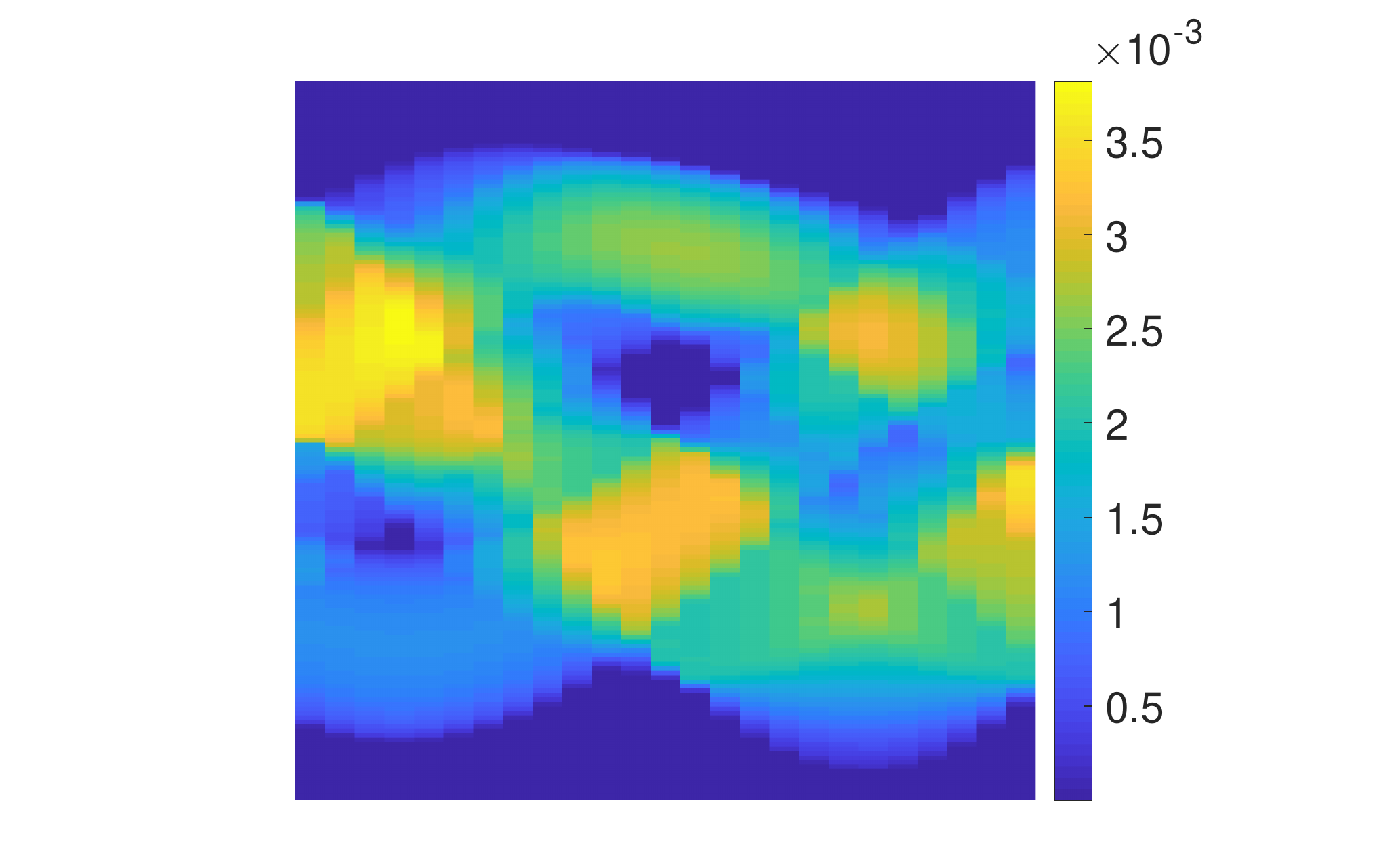}
	\end{subfigure}
	\caption{Reconstruction of phantom and sinogram for sequential line rotation via $\mathcal{M}_1.$ ($\alpha_1=2\cdot 10^4,\; \alpha_2=0.1^{5.1},\; \text{SSIM}\brackets{\mathbf{c}} = 0.9492$)}
	\label{Fig:Reco_discr_discr}
\end{figure}
Next, we regard results for the setting of simultaneous line rotation. Similar to~\cite{bringout2020new}, we get a slightly rotated version of the phantom when ignoring the different sampling pattern, see Figure~\ref{Fig:Reco_Phantom_discr_cont}.
\begin{figure}[htbp]
	\centering
	\includegraphics[width=0.45\linewidth, trim=0 8.25cm 0 8cm, clip]{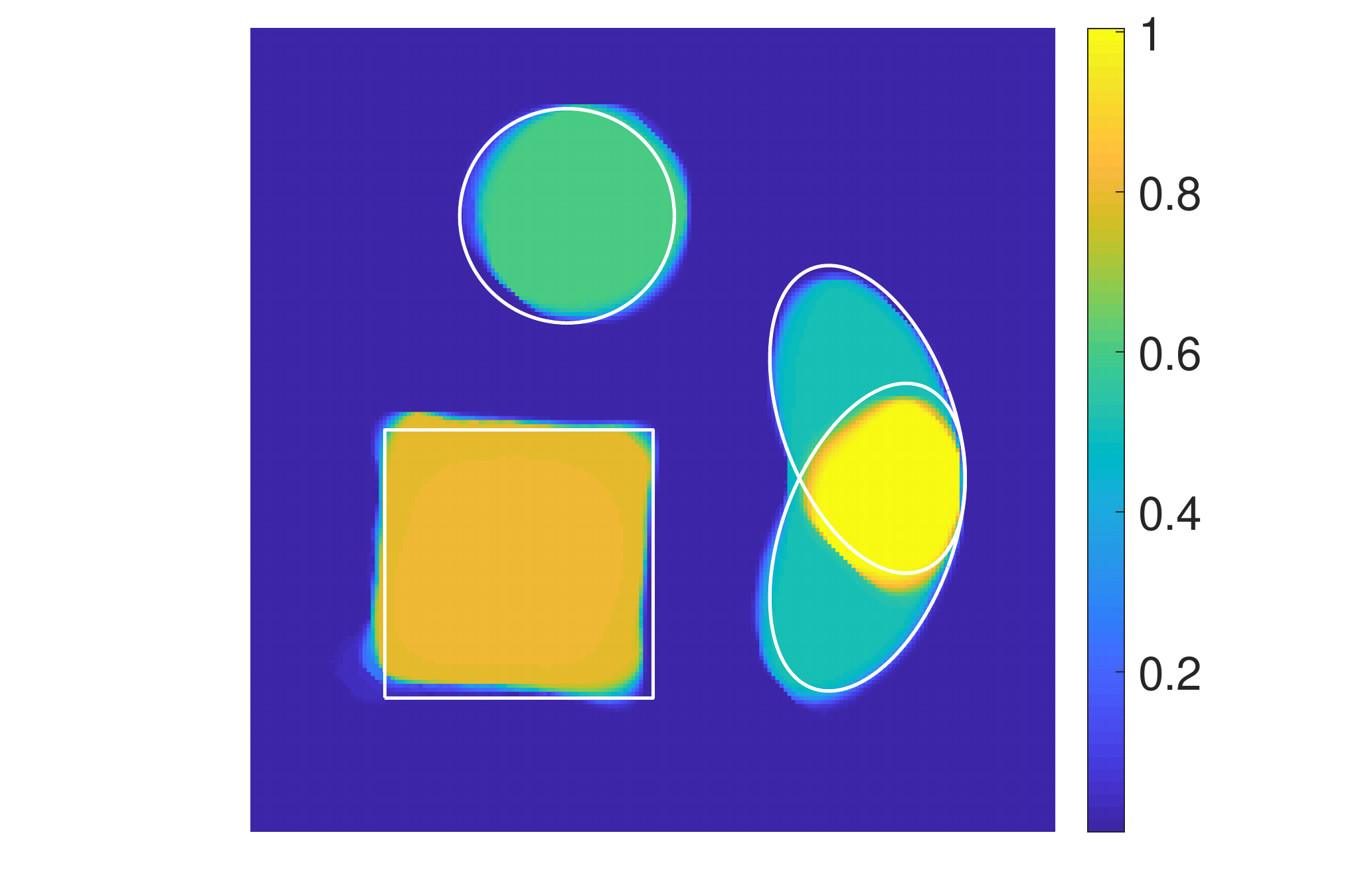}
	\caption{Phantom reconstruction for simultaneous line rotation using model  $\mathcal{M}_1$, i.e. ignoring the specific sampling pattern for simultaneous line rotation. Contours of the groundtruth are depicted in white. ($\alpha_1=4\cdot 10^4,\; \alpha_2=0.1^{4},\; \text{SSIM}\brackets{\mathbf{c}} = 0.7976$)}
	\label{Fig:Reco_Phantom_discr_cont}
\end{figure}\\
Lastly, we compare reconstructions incorporating $\mathcal{K}_{3,l} \mathcal{R}$ with those neglecting this term. According to Figure~\ref{Fig:IntegralEstimates_3} we expect only slight differences. Corresponding deliverables with and without added gaussian noise with $2\cdot 10^{-10}$ standard deviation (ca. 0.79\% of $u^*$) can be found in Figure~\ref{Fig:Recos}. Our reconstructions show that incorporation of the third term in the model yields slightly higher values for structural similarity and for the noiseless case, the quantitative concentration values fit better to the groundtruth. For noisy data the size of $\mathcal{K}_{3,l} \mathcal{R}$ gets close to the range of added noise and thus, it is reasonable that the gain in the SSIM value is lower compared to the noise-free setting.
\begin{figure}[htb]
	\centering
	\begin{subfigure}[b]{\textwidth}
		\centering
		\includegraphics[width=0.425\linewidth, trim=2cm 9cm 0 8cm, clip]{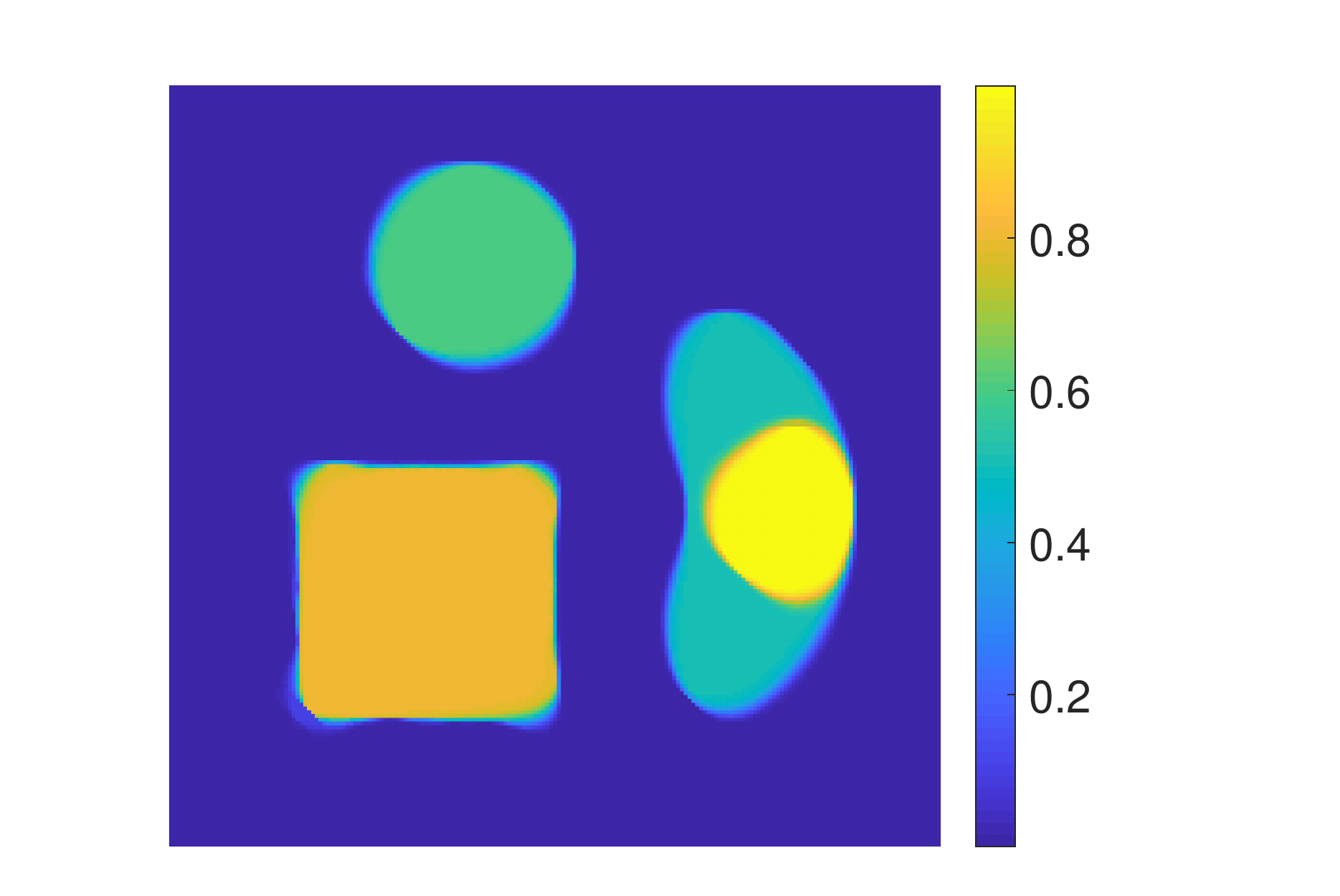}%
		\hfill
		\includegraphics[width=0.425\linewidth, trim=2cm 9cm 0cm 8cm, clip]{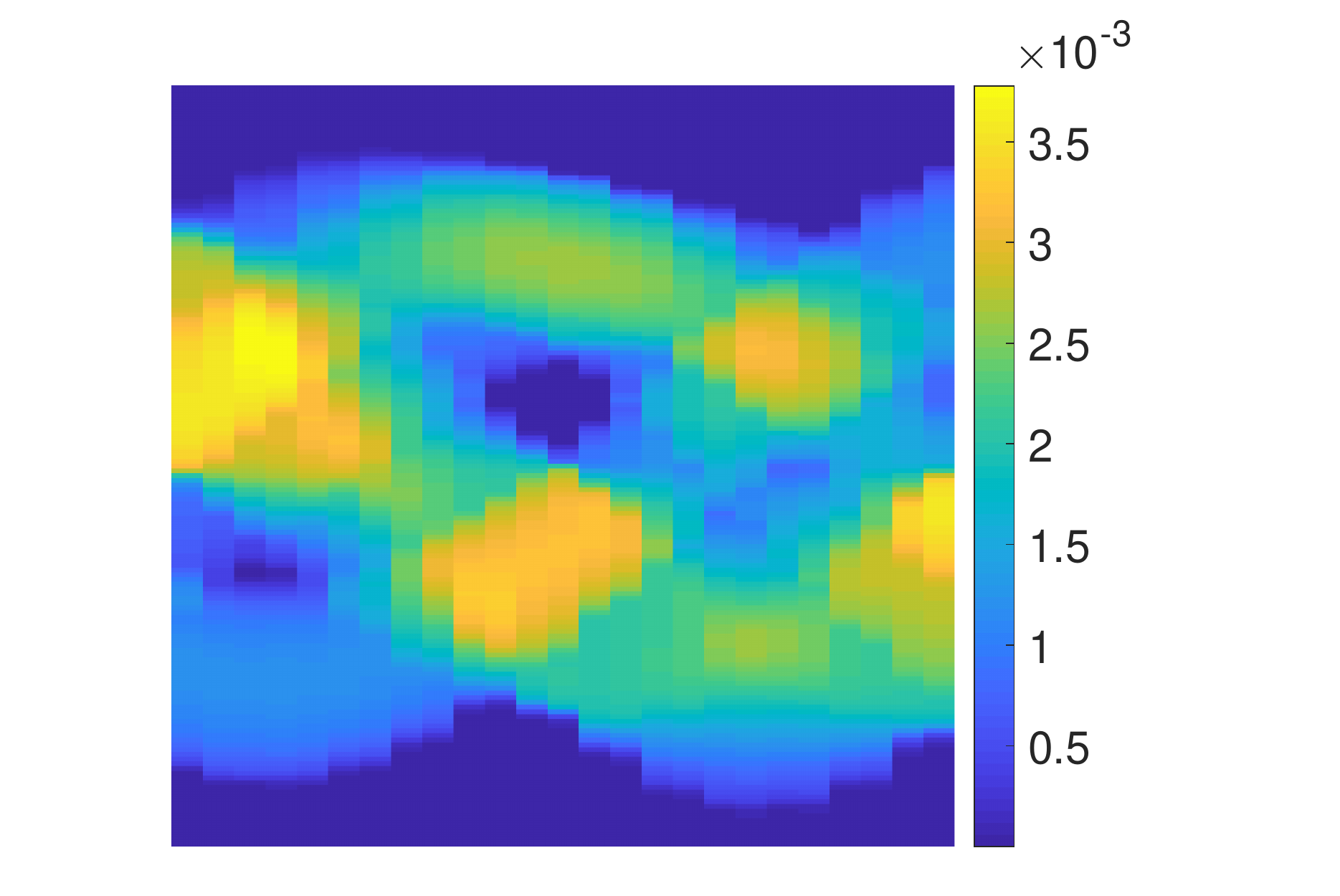}
		\caption{Reconstruction with respect to $\mathcal{M}_2.$ ($\alpha_1=4\cdot 10^4,\; \alpha_2=0.1^{3.75},\; \text{SSIM}\brackets{\mathbf{c}} = 0.9014$)}
	\end{subfigure}
	\vskip\baselineskip
	\begin{subfigure}[b]{\textwidth}
		\centering
		\includegraphics[width=0.425\linewidth, trim=2cm 9cm 0 8cm, clip]{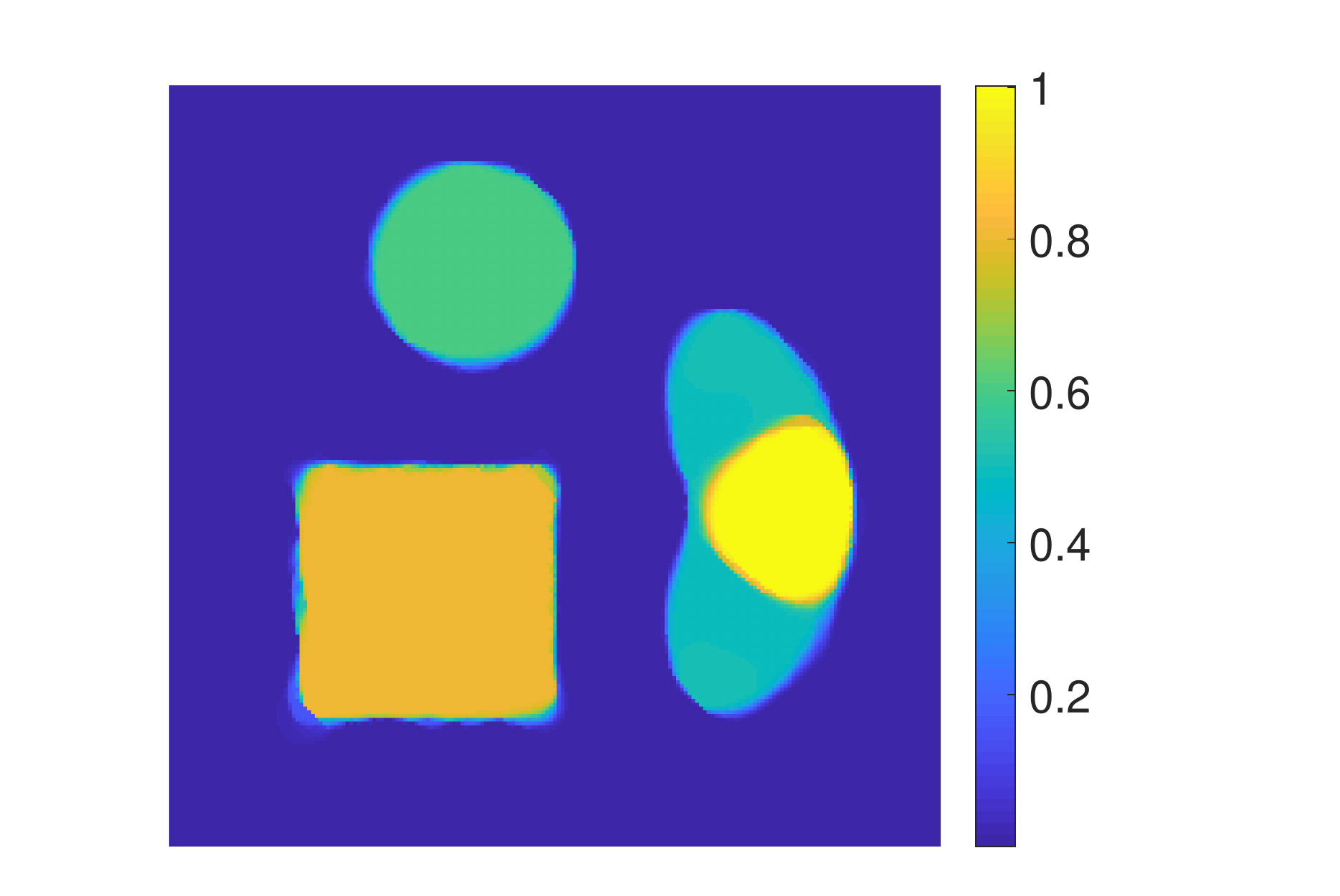}%
		\hfill
		\includegraphics[width=0.425\linewidth, trim=2cm 9cm 0cm 8cm, clip]{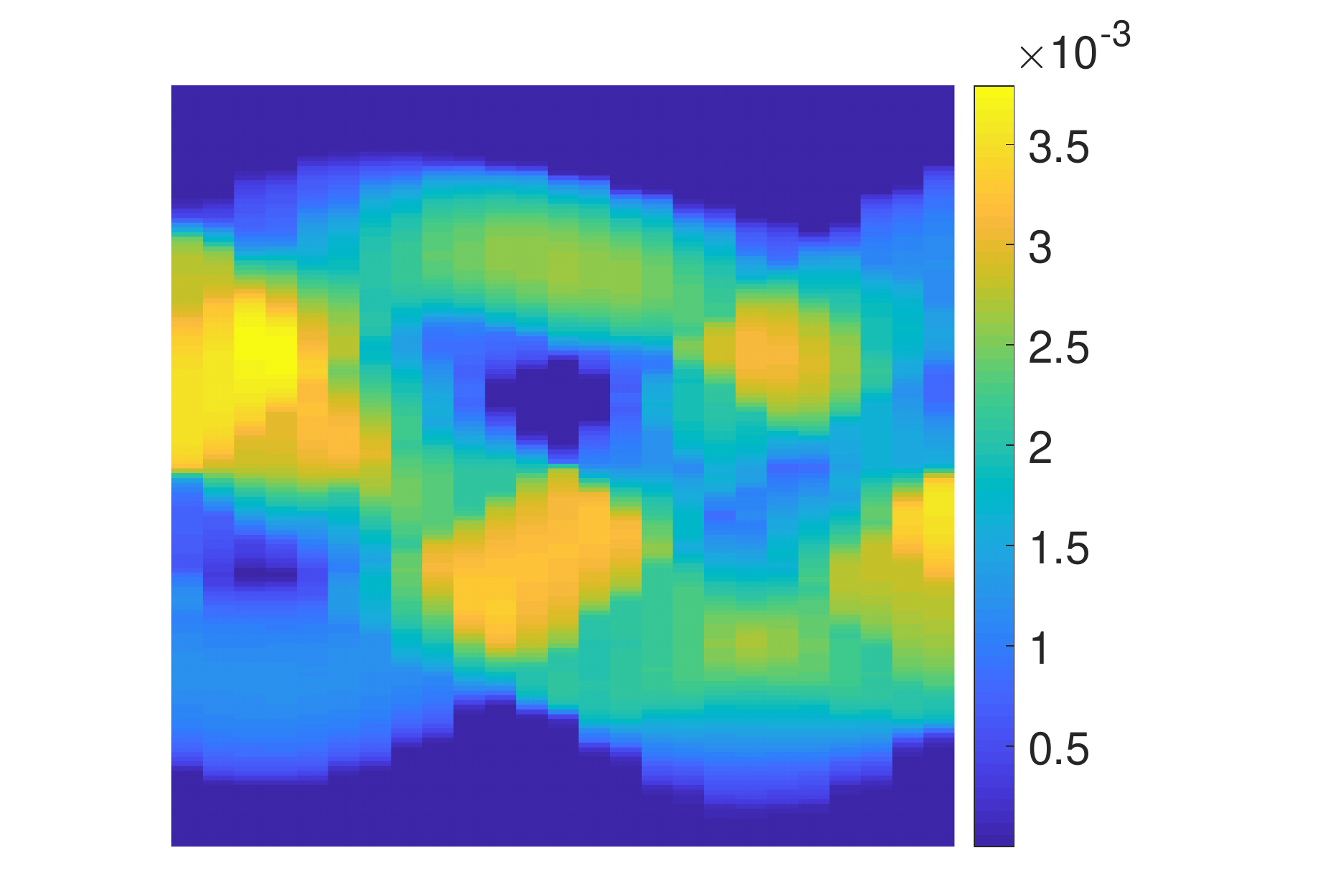}
		\caption{Reconstruction with respect to $\mathcal{M}_3.$ ($\alpha_1=1\cdot 10^4,\; \alpha_2=0.1^{4.7},\; \text{SSIM}\brackets{\mathbf{c}} = 0.9302$)}
	\end{subfigure}
	\vskip\baselineskip
	\begin{subfigure}[b]{\textwidth}
		\centering
		\includegraphics[width=0.425\linewidth, trim=2cm 9cm 0 8cm, clip]{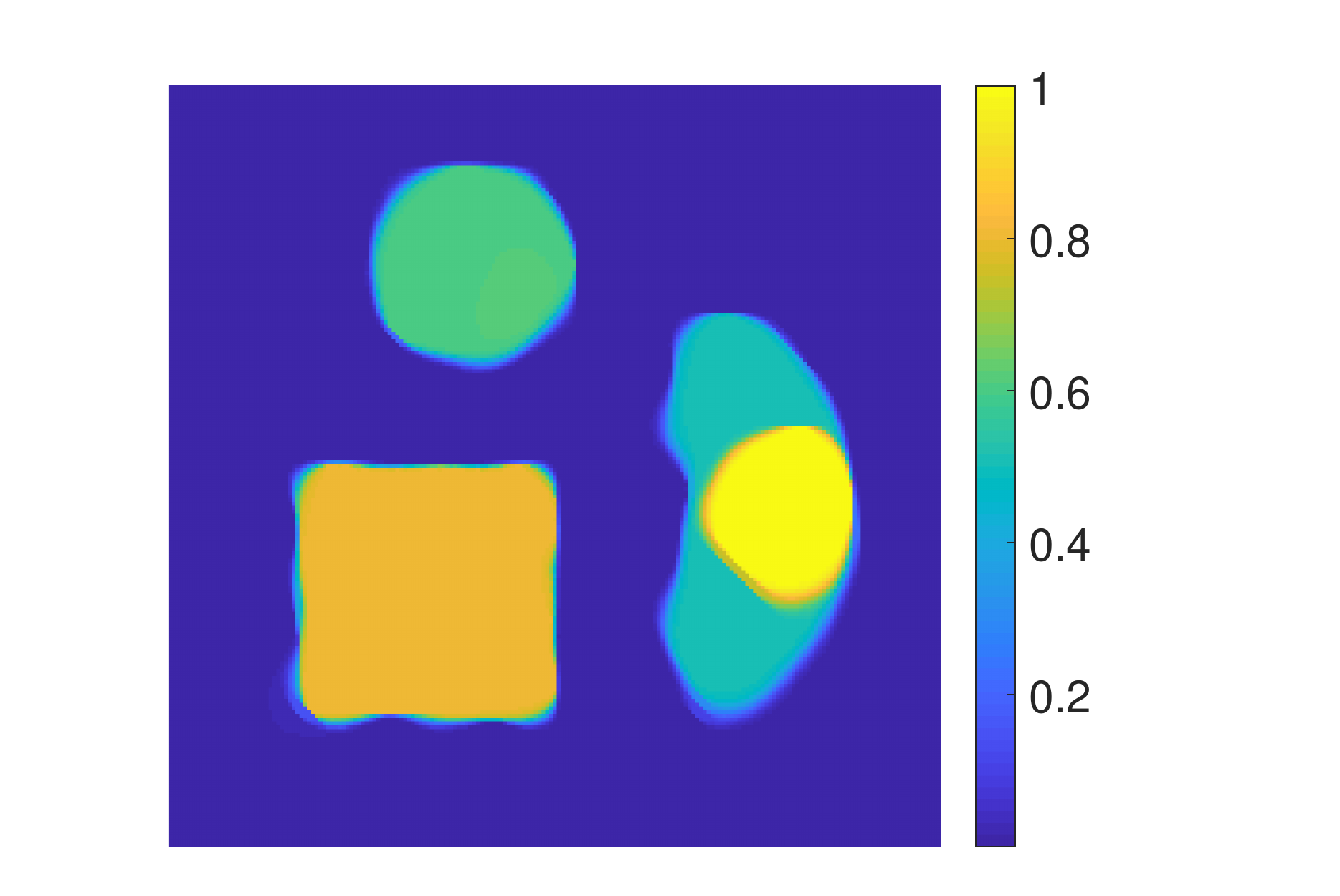}%
		\hfill
		\includegraphics[width=0.425\linewidth, trim=2cm 9cm 0cm 8cm, clip]{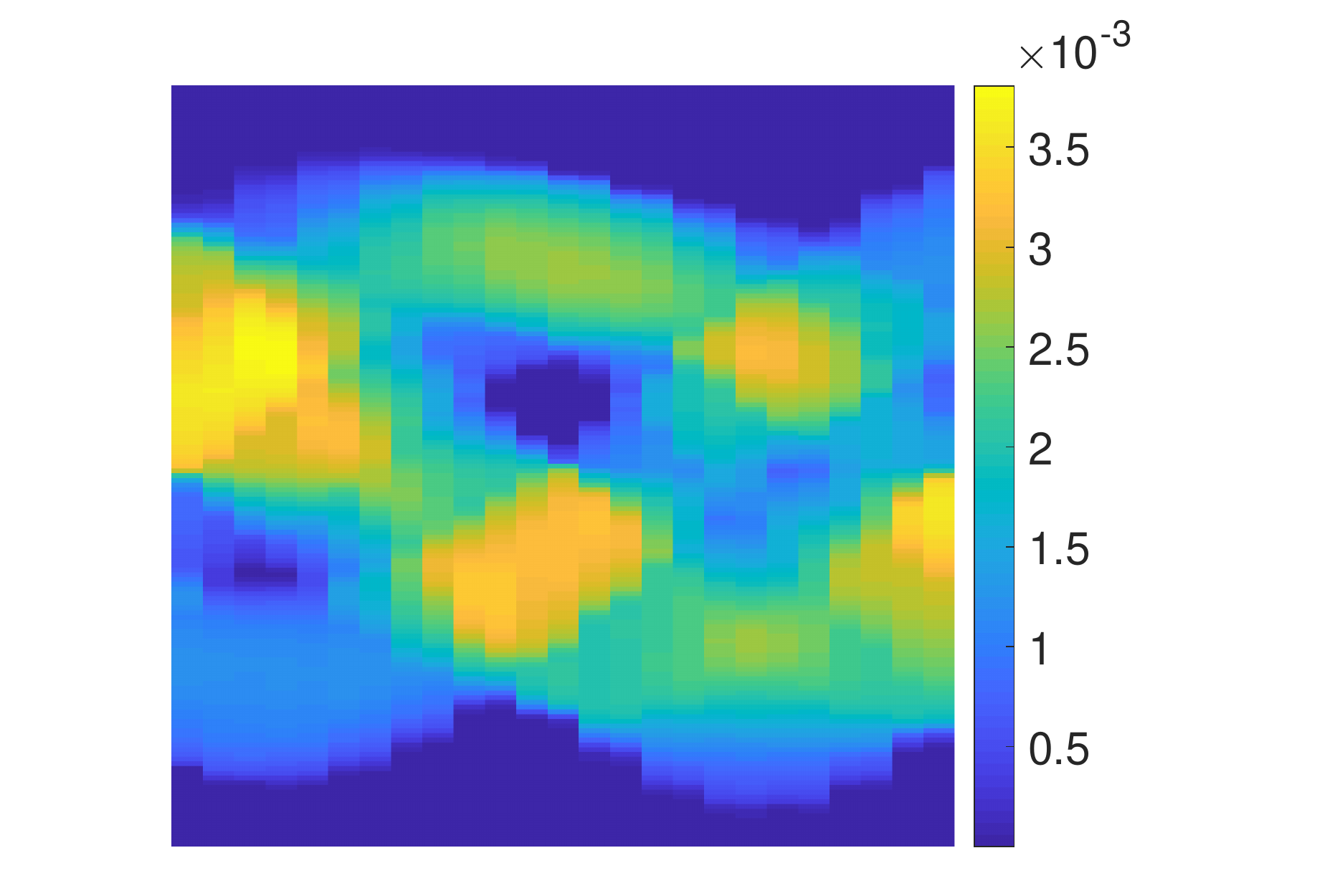}
		\caption{Reconstruction with respect to $\mathcal{M}_2$ and noisy data. ($\alpha_1=4\cdot 10^4,\; \alpha_2=0.1^{3.8},\; \text{SSIM}\brackets{\mathbf{c}} = 0.8689$)}
	\end{subfigure}
	\vskip\baselineskip
	\begin{subfigure}[b]{\textwidth}
		\centering
		\includegraphics[width=0.425\linewidth, trim=2cm 9cm 0 8cm, clip]{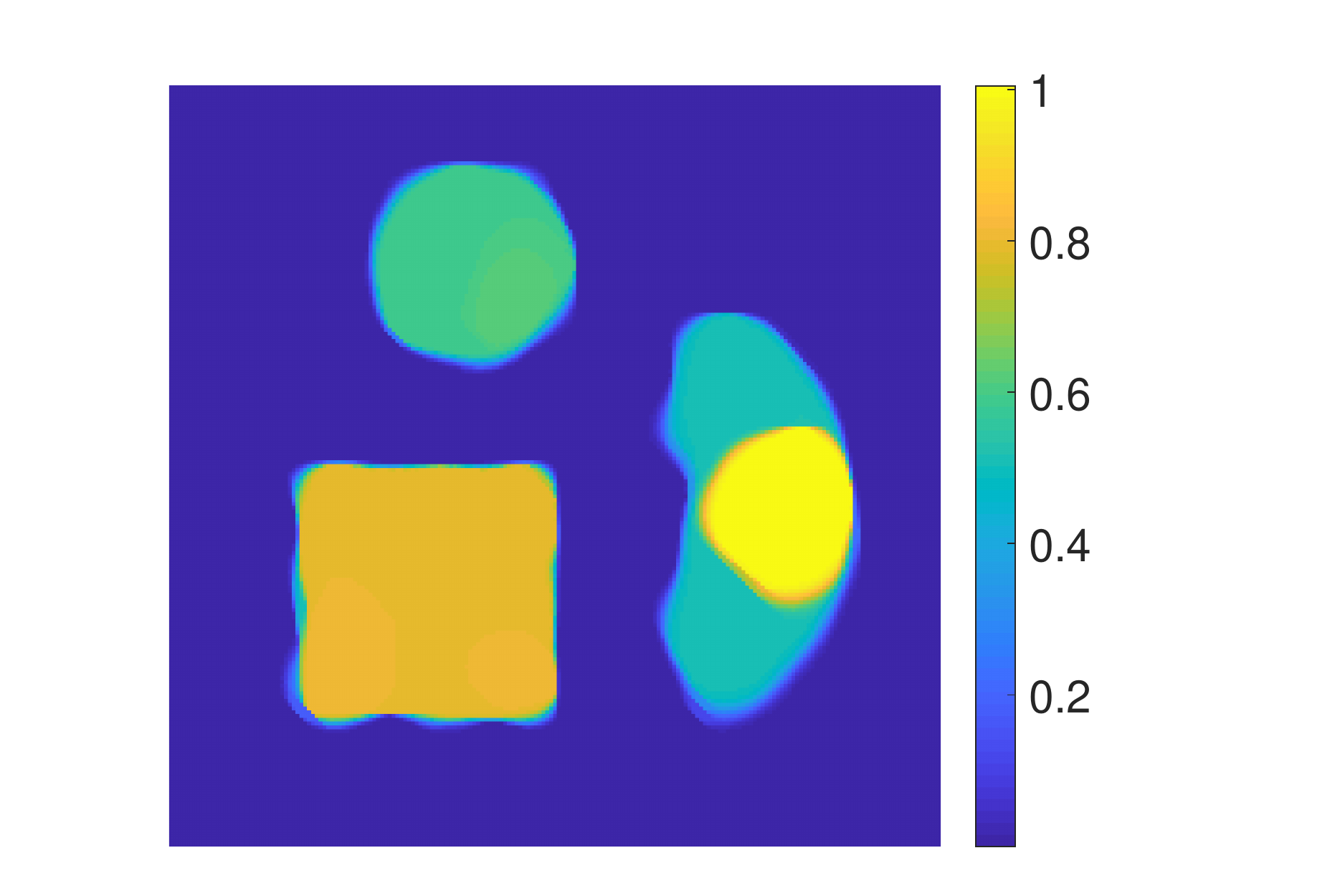}%
		\hfill
		\includegraphics[width=0.425\linewidth, trim=2cm 9cm 0cm 8cm, clip]{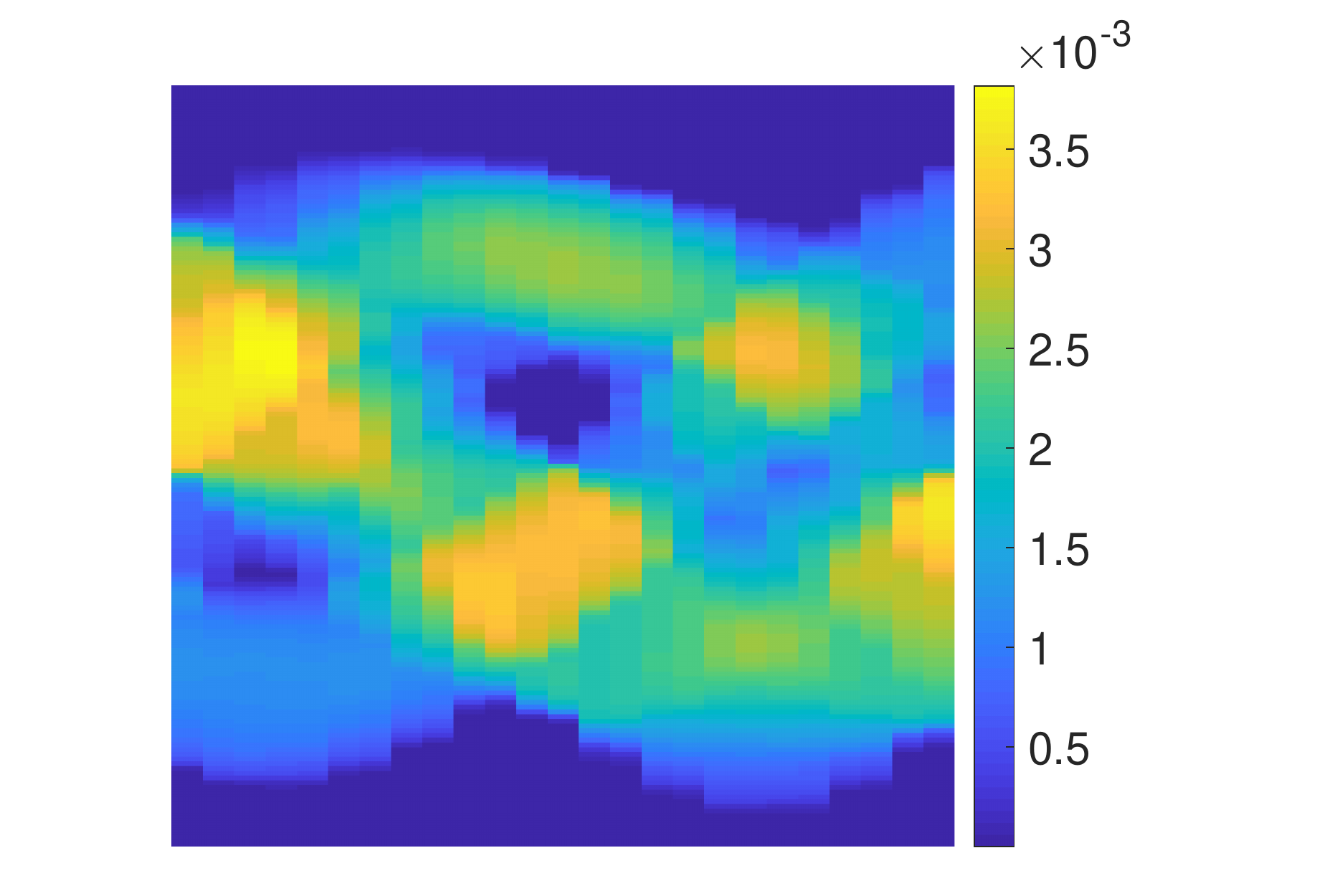}
		\caption{Reconstruction with respect to $\mathcal{M}_3$ and noisy data. ($\alpha_1=4\cdot 10^4,\; \alpha_2=0.1^{3.9},\; \text{SSIM}\brackets{\mathbf{c}} = 0.8826$)}
	\end{subfigure}
	\caption{Results for simultaneous line rotation using different reconstruction approaches.}
	\label{Fig:Recos}
\end{figure}

%%%%%%%%%%%%%%%%%%%%%%%%%%%%%%%%%%
\section{Conclusion}
%%%%%%%%%%%%%%%%%%%%%%%%%%%%%%%%%%

In this work, we regarded the relation of MPI and Radon data for a continuously rotating FFL. Because of the time derivative in the signal equation and the additional time-dependencies for this setting, we get two additive terms in the forward model. We derived bounds for these two terms supporting the assumption of \cite{bringout2020new} and~\cite{knopp2011fourier} that for sufficiently slow FFL rotation compared to the translation speed, the results from the sequential line rotation setting can be applied. This is also emphasized by our numerical results. Incorporation of the additional terms for image reconstruction does not seem to be necessary on a first view. However, including the third term in our reconstruction approach is simple because both terms together still form a convolution with the Radon data of the particle concentration. Further, its bound can be computed beforehand of reconstruction and thus, its magnitude can be compared to the total signal. In our case, incorporation yielded slightly better structural similarity values. Future research will comprise steps towards a more realistic set of assumptions, like time-dependent particle concentrations, magnetic field imperfections, and that the direct feedthrough of the excitation signal needs to be removed. Moreover, it would be interesting to extend the results of this article to different modeling approaches distinct from the Langevin model.

%%%%%%%%%%%%%%%%%%%%%%%%%%%%%%%%%%
\section*{Acknowledgment}
The authors acknowledge the support by the Deutsche Forschungsgemeinschaft (DFG) within the Research Training Group GRK 2583 "Modeling, Simulation and Optimization of Fluid Dynamic Applications".
%%%%%%%%%%%%%%%%%%%%%%%%%%%%%%%%%%

%%%%%%%%%%%%%%%%%%%%%%%%%%%%%%%%%%
% References
%%%%%%%%%%%%%%%%%%%%%%%%%%%%%%%%%%

\bibliographystyle{plain}
\bibliography{literature}

\end{document}